\documentclass[11pt]{article}

\usepackage{tikz}
\usetikzlibrary{matrix,arrows,shapes}
\usepackage{pgf}
\usepackage{url}
  \usepackage{enumerate}
  \usepackage{amsthm,amssymb,amsmath,amscd}
    \usepackage[retainorgcmds]{IEEEtrantools}
  \usepackage{eucal}
  \usepackage{mathrsfs}
  \usepackage{mathtools}
  \usepackage{upgreek}
   \usepackage{dsfont}
      \usepackage{yfonts}
  \usepackage[T1]{fontenc}
  \usepackage{bbm}
  \usepackage{geometry}
\usepackage{array}
\usepackage{booktabs}
\usepackage[retainorgcmds]{IEEEtrantools}
  \usepackage{ulem}
\newcommand{\cD}{\mathcal{D}}
\newcommand{\cE}{\mathcal{E}}
\newcommand{\cH}{\mathcal{H}}

\newcommand{\cL}{\mathcal{L}}

\newcommand{\cO}{\mathcal{O}}

\newcommand{\Om}{\Omega}
\newcommand{\ep}{\epsilon}
\newcommand{\x}{\mathbf{x}}

\newcommand{\E}{\mathcal{E}}

\newcommand{\fA}{\mathfrak{A}}

\newcommand{\F}{\mathfrak{F}}

\newcommand{\frakg}{\mathfrak{g}}

\newcommand{\Gcal}{\mathcal{G}}  
\newcommand{\Lcal}{\mathcal {L}}
\newcommand{\Bcal}{\mathcal {B}}
  
\newcommand{\Pcal}{\mathcal{P}}  
\newcommand{\Hcal}{\mathcal{H}}  
\newcommand{\Ccal}{\mathcal{C}}
\newcommand{\Dcal}{\mathcal{D}}
\newcommand{\Ecal}{\mathcal{E}} 
\newcommand{\Fcal}{\mathcal{F}}

\newcommand{\Mcal}{\mathcal{M}}
\newcommand{\Ocal}{\mathcal{O}}
\newcommand{\Scal}{\mathcal{S}}
\newcommand{\Rcal}{\mathcal{R}}
\newcommand{\Tcal}{\mathcal{T}}
\newcommand{\Vcal}{\mathcal{V}}
\newcommand{\Wcal}{\mathcal{W}}

\newcommand{\Zcal}{\mathcal{Z}}
\newcommand{\Ycal}{\mathcal{Y}}

\newcommand{\Ci}{\mathcal{C}^\infty} 
\newcommand{\sst}[1]{\scriptscriptstyle{#1}}  
\newcommand{\vr}[1]{\boldsymbol{#1}}         
\newcommand{\minus}{\sst{-1}}   
\newcommand{\1}{\mathds{1}}                         
\newcommand{\pa}{\partial}                              

\newcommand{\NN}{\mathbb{N}}          
\newcommand{\RR}{\mathbb{R}}           
\newcommand{\CC}{\mathbb{C}}           
\newcommand{\MM}{\mathbb{M}} 	     
\newcommand{\al}{\alpha}
\newcommand{\bet}{\beta}

\newcommand{\Ga}{\Gamma}
\newcommand{\de}{\delta}

\newcommand{\la}{\lambda}
\newcommand{\La}{\Lambda}

\newcommand{\ph}{\varphi}
\newcommand{\om}{\omega}
\newcommand{\T}{\cdot_{{}^\Tcal}}

\newcommand{\TT}{\Tcal}

\newcommand{\Poi}[2]{\{#1,#2\}}

\newcommand{\WF}{\mathrm{WF}}         
\newcommand{\id}{\mathrm{id}}               
\DeclareMathOperator{\supp}{\mathrm{supp}}      
\newcommand{\pp}{\mathrm{pp}}
\newcommand{\rp}{\mathrm{rp}}   
\newcommand{\MS}{\mathrm{MS}}
\newcommand{\loc}{\mathrm{loc}}
\newcommand{\be}{\begin{equation}}
\newcommand{\ee}{\end{equation}}
\newcommand{\reg}{\mathrm{reg}}
\newcommand{\mc}{\mu\mathrm{c}}
\def\normOrd#1{\mathop{:}\nolimits\!#1\!\mathop{:}\nolimits}
\newcommand{\DIAG}{\mathrm{DIAG}} 	
\newcommand{\sd}{\mathrm{sd}}
\DeclareMathOperator{\cha}{\textrm{char}}

\begin{document}

\title{Perturbative Construction of Models of Algebraic Quantum Field Theory}
\author{\null\\ Klaus Fredenhagen$^{(1)}$, Katarzyna Rejzner$^{(2)}$ \\
  \null\\
  \null\\\small{$^{(1)}$ II Institute for Theoretical Physik, University of Hamburg,}\\
\small{$^{(2)}$ Department of Mathematics, University of York}\\
\small{\texttt{klaus.fredenhagen@desy.de,} \texttt{kasia.rejzner@york.ac.uk }}}
%
%
\maketitle

  \theoremstyle{plain}
  \newtheorem{definition}{Definition}
  \newtheorem{theorem}{Theorem}
  \newtheorem{proposition}{Proposition}
  \newtheorem{cor}{Corollary}
  \newtheorem{lemma}{Lemma}
  \newtheorem{exa}{Example}
  
  \theoremstyle{definition}
  \newtheorem{rem}{Remark}
 \theoremstyle{definition}
  \newtheorem{ass}{\underline{\textit{Assumption}}}

\abstract{We review the construction of models of algebraic quantum field theory by renormalized perturbation theory.}

\section{Introduction}
\label{Introduction}

The axiomatic framework of AQFT allows for a qualitative description of a large class of phenomena occurring in particle physics and some parts of solid state physics. It does not, however, yield quantitative predictions, and there is a widespread impression that one has to abandon the formalism of AQFT if one wants to make real contact with experiments. Actually, up to now no single model of an interacting AQFT in 4d Minkowski space has been constructed.

But what are the alternatives? Standard textbooks on QFT either start from canonical quantization of free field theory on Fock space and try to construct the interacting theory in the interaction picture, or they use the path integral formalism. The canonical approach ends up in the Gell-Mann Low formula for the vacuum expectation values of time ordered products of fields,
\begin{equation}
\om_0(T\ph(x_1)\dots\ph(x_n))=\frac{\langle\Om,T\ph_0(x_1)\dots\ph_0(x_n)e^{\frac{i}{\hbar}\int \mathcal L_I(x)d^4x}\Om\rangle}{\langle\Om,Te^{\frac{i}{\hbar}\int \mathcal L_I(x)d^4x}\Om\rangle}\ ,
\end{equation}
where $\ph_0$ is the free field treated as an operator valued distribution on the Fock space, $\mathcal L_I$ is the interaction density treated as a Wick polynomial of $\ph_0$ and $\Omega$ is the vacuum vector of the free theory. The time ordering symbol $T$ means that the products have to be performed after ordering of the factors according to their time arguments. 

The path integral approach reinterprets the Gell-Mann Low formula as an integral over all classical field configurations $\phi$
\begin{equation}
\om_0(T\ph(x_1)\dots\ph(x_n))=Z^{-1}\int \,\phi(x_1)\dots\phi(x_n)e^{\frac{i}{\hbar}\int \mathcal L(x)d^4x}D\phi 
\end{equation}
 where now $\mathcal L$ is the full classical Lagrangian, $Z$ is a normalization factor, and $D\phi$ is thought of as the Lebesgue integral over field space.
 
 Both versions are only heuristic, and it required the hard and ingenious work of several generations of physicists to turn these formal expressions into unambiguous computations.
 The state of the art is that one can create a formal power series in $\hbar$ where every term is well defined, up to some remaining infrared problems originating from the integral over Minkowski space in the exponent. The great success of QFT relies on the fact that already the first few terms of this series yield a good and often even excellent agreement with experimental data.
 
 The path integral approach has the advantage that it is formally similar to probability theory. Actually, by passing to imaginary time (Wick rotation), one can interpret the vacuum expectation values of time ordered products of fields as correlation functions of a probability distribution (euclidean QFT). In particular, counter-intuitive properties of quantum physics as e.g. entanglement do not occur. Moreover, the momentum space integrals in the evaluation of Feynman diagrams have better convergence properties. Finally, due to the Osterwalder-Schrader theorem, a Wick rotation back to real time is possible under very general conditions.
 
The disadvantage of the path integral approach is that the noncommutative product of operators, which is crucial for the structure of quantum physics, appears only indirectly in terms of different boundary values of analytic functions. In the canonical approach, the operator product is given from the beginning, but there the definition of the time ordered product is problematic.  
First of all, it is not well defined as a product of operators, since, by the existence of a deterministic time evolution, fields at a given time can be expressed in terms of fields at an earlier time, and thus the time ordering prescription is ambiguous. One may instead define time ordered products $TA(t_1)\dots A(t_n)$ of  an operator valued function of time $t\mapsto A(t)$ as a symmetric operator valued function of $n$ time variables such that 
\begin{equation}
TA(t_1)\dots A(t_n)=A(t_1)\dots A(t_n)\  {\rm if }\ t_1\ge\dots\ge t_n\ .
\end{equation}
This, however, does not work since the quantum fields are distributions, and the time ordering prescription would amount to multiply them with a discontinuous function. 

But there is a way out, as first observed by St\"uckelberg, further elaborated by Bogoliubov and collaborators and finally worked out by Epstein and Glaser (causal perturbation theory). Namely, one may define the time ordered product of $n$ fields as an operator valued distribution which is already known for non-coinciding points. Due to the UV divergences of QFT, the extension to coinciding points is ambiguous, but the crucial observation is that this ambiguity is the same ambiguity which occurs in the removal of infinities in approaches where the theory is regularized by the 
introduction of a momentum cutoff, and where the theory without cutoff has to be fixed by renormalization conditions.

Originally, the insertion of a test function $g$ into the interaction Lagrangian, instead of integrating it over all spacetime, was considered to be an intermediate step, and in the last step one aimed at the limit where $g$ tends to 1 (adiabatic limit). In this limit one then finds vacuum expectation values of operator products of time ordered products of interacting fields, and using the Wightman reconstruction theorem, one obtains interacting fields as operator valued distributions on some ``Hilbert space'', of course only in the sense of formal power series. But as first observed in \cite{IS78} and rediscovered in \cite{BF00} the algebra of observables associated to some bounded region can be already constructed if one chooses a test function $g$ which is equal to 1 on some slightly larger region. Actually, the full Haag-Kastler net of the interacting theory can be obtained in this way. Thus causal perturbation theory provides a direct way for a construction of the algebra of observables. Hence by replacing the condition that the local algebras have to be unital C*-algebras by the condition that they are isomorphic to unital *-algebras of formal power series of operators on a dense invariant subspace of some Hilbert space, one obtains a huge class of models, in particular the models used in elementary particle physics.

On this level, structural properties of the local net can be analyzed, but the powerful structural results on C*- and von Neumann algebras are not available. 
Nevertheless, one can derive interesting results, as e.g. the validity of the time-slice axiom \cite{CF08}, the existence of operator expansions \cite{Hollands07} and an algebraic version of the Callan-Symanzik equation \cite{BDF09}. 

In order to reach numerical predictions, one needs in the next step a construction of states. States are here defined as linear maps from the algebra to the formal power series over $\CC$, and the positivity condition on states now means that the expectation value of $A^*A$ is the absolute square of another power series. The construction of states can be done via the adiabatic limit as described above; this is the way the vacuum state is constructed in \cite{EG}. As observed by Steinmann \cite{Steinmann95}, this method does not work for the construction of KMS states. The reason is that the analog of the Gell-Mann Low formula does not hold at nonzero temperature, due to the different asymptotic time behavior of free and interacting systems at nonzero temperature. But here a structural result helps: namely,  the time-slice axiom allows to treat only the theory within a short time interval, and the asymptotic behavior in time does not matter for the existence of states. What matters is the decay of correlations in spacelike directions which is exponentially fast for massive theories.  
 
Up to now we considered the so-called on shell formalism, due to the fact that we constructed the operators on Fock space, thereby imposing the validity of the Klein Gordon equation for the free field. It turned out, however, to be more useful to replace Fock space operators by functionals of classical field configurations which are not restricted to those which satisfy the field equation. On the space of functionals one can then introduce several operations: the pointwise (classical) product, the involution by complex conjugation, the Peierls bracket (as a covariant version of a Poisson bracket on the space of functionals), the non-commutative, associative 
$\star$-product in the sense of deformation quantization and the time ordered product. It is the latter which is relevant for inducing the interaction and which requires renormalization.
The other operations can be directly defined. This is trivial for the pointwise product and for the involution. The Peierls bracket is obtained by considering the linearized Euler-Lagrange operator, which e.g. for the $\ph^4$-theory looks like
\begin{equation}
E'(\ph)=\square +m^2+\frac{\lambda}{2} \ph^2 
\end{equation}
where the last term acts as a multiplication operator. We consider only theories where the linearized Euler-Lagrange operator is normally hyperbolic and hence has unique retarded and advanced propagators $\Delta^{R/A}(\ph)$.
The Peierls bracket is then defined by
\begin{equation}\label{Peierls2}
\{F,G\}(\ph)=\left\langle\frac{\delta F} {\delta \ph}(\ph),\Delta(\ph)\frac{\delta G}{\delta \ph}(\ph)\right\rangle
\end{equation}
where $\Delta=\Delta^R-\Delta^A$ and $\frac{\delta}{\delta \ph}$ is the functional derivative (defined as the directional derivative). In free theories, $E'$ and then also the propagators do not depend on $\ph$. One then can define the $\star$-product (in the sense of formal power series in $\hbar$) by
\begin{equation}
(F\star G)(\ph)=\left.e^{\frac{i\hbar}{2}\left\langle\frac{\delta}{\delta \ph},\Delta \frac{\delta}{\delta\ph'}\right\rangle}F(\ph)G(\ph')\right|_{\ph'=\ph} \ .
\end{equation}
The time-ordered product is defined by a similar formula
\begin{equation}\label{tordered}
(F\T G)=\left.e^{\hbar\left\langle\frac{\delta}{\delta \ph},\Delta^D \frac{\delta}{\delta\ph'}\right\rangle}F(\ph)G(\ph')\right|_{\ph'=\ph}
\end{equation}
with the Dirac propagator $\Delta^D=\frac12(\Delta^R+\Delta^A)$.

Note that there is a crucial difference between the time ordered product and the other products. Namely, the ideal generated by the field equation with respect to the pointwise product is also an ideal with respect to the Poisson bracket and the $\star$-product, but not with respect to the time ordered product. This is actually a necessary condition which allows to use the time ordered product for introducing an interaction. Let $V$ be the interaction. We then define  the interacting observables by
\begin{equation}
R_V(F)=(e_{\Tcal}^V)^{\star -1}\star (e_{\Tcal}^V\T F)\ .
\end{equation}
Here $e_{\Tcal}$ means the exponential series where powers are computed via the time ordered product. For an evaluation functional $\Phi_x(\ph)=\ph(x)$, the corresponding interacting field $x\to R_V(\Phi_x)$ satisfies the equation
\begin{equation}
E'R_V(\Phi_x)=E'\Phi_x+R_V\left(\frac{\delta V}{\delta\ph(x)}\right)
\end{equation}
which may be interpreted as the field equation with interaction $V$, when evaluated on some $\ph$ which satisfies the free field equation. 

The rough description of the formalism has to be made precise in the following sense: One has to specify the functionals which are allowed, and one has to check whether this class contains the relevant ones.
As it stands we need functionals whose functional derivatives are test functions in order that all operations are well defined. But
we will see that by changing the products to equivalent ones which corresponds to Wick ordering in the Fock space framework one can extend the Peierls bracket and the $\star$-product to a rather large class of functionals, which contains in particular the local functionals that appear as terms in the Lagrangian, and is stable under these operations. The definition of time ordered products is more involved and there are different possibilities, corresponding to the choice of renormalization conditions.

The plan of the paper is as follows: we will first outline the functional analytic tools which are needed for the operations. We then define the Peierls bracket and the $\star$-product on a class of functionals called microcausal.
Thereafter we come to the problem to define the time ordered products. Here we first develop the general formalism and show that it leads to a construction of local nets. We review some structural properties of these nets, in particular their behavior under renormalization group transformations. Finally, we outline a possible construction of states. 

\section{Functional derivatives, wave front sets, and all that}
\label{sec:Functional derivatives}

Our approach to quantum field theory is based, in this respect similar to the path integral approach, on functionals of classical field configurations. But there, at least in its euclidean version, the measure theoretic aspects of the space of field configurations are of central importance; in our case, due to the frequent use of functional derivatives, the properties of the space of field configurations as a differential manifold are crucial. 

For definiteness we concentrate on the case of a scalar field and fix an oriented, time-oriented globally hyperbolic spacetime $M$. There we consider the space of real valued smooth functions as the space of field configurations,
\begin{equation}
\cE=\mathcal C^{\infty}(M,\RR)\ .
\end{equation}
We model it as a differentiable manifold over the space of compactly supported smooth functions
\begin{equation}
\cD=\mathcal C^{\infty}_c(M,\RR)\ 
\end{equation}
where charts are defined as maps
\begin{equation}
\ph+\cD\to \cD,\ \ph+\vec{\ph}\mapsto \vec{\ph}\ , {\rm with }\ \ph\in\cE\ .
\end{equation}
Clearly, $\cE$ has the structure of an affine manifold. A similar affine structure can be introduced also for other fields, as e.g. gauge theories or gravity.

We will model observables as functionals on $\Ecal$ and we allow only functionals which depend on the field configuration inside some compact region. This includes in particular the polynomial functionals
\begin{equation}
F(\ph)=\sum_{k=1}^n\int \ph(x_1)\dots\ph(x_k)f_k(x_1,\dots,x_k)
\end{equation}
with symmetrical distributional densities $f_k$ with compact support. More generally, we consider functionals, for which all functional derivatives $ F^{(n)}$ exist and are continuous. We recall after \cite{Ham} (see \cite{Neeb} for a review) that a functional derivative of a functional is defined as
\begin{equation}
\langle F^{(1)}(\ph),\vec{\ph}\rangle:=\frac{d}{d\lambda}F(\ph+\lambda{\vec{\ph}})\big|_{\la=0}
\end{equation}
and a functional is differentiable if the derivative exists for all $\ph\in\cE$. It is continuously differentiable if  the map $\cE\times\cD\to \CC, (\ph,\vec{\ph})\mapsto \langle F^{(1)}(\ph),\vec{\ph}\rangle$ is continuous. If $\Ecal$ is taken with its natural Fr\'echet topology, this implies that $F^{(1)}(\ph)$ is a compactly supported distributional density.
Higher derivatives are obtained by iterating this definition, i.e.
\be
F^{(n)} (\ph)(\vec{\ph}_1 , \ldots , \vec{\ph}_n ) := \frac{d}{d\la} F^{(n-1)} (\ph + \la\vec{\ph}_n )(\vec{\ph}_1 , \ldots, \vec{\ph}_{n-1} ) 
\big|_{\la=0}\ ,
 \ee
  and we find that the $F^{(n)}(\ph)$'s are symmetric compactly supported distributional densities with support contained in $K^n$, for some compact set $K\subset M$. There remains, however, the problem that the propagators have singularities, and therefore the contractions with the distributional densities occurring as functional derivatives are not always well defined. The restriction to functionals whose functional derivatives are smooth densities, on the other side, would exclude almost all local functionals, i.e. functionals of the form
\begin{equation}
F(\ph)=\int f(j_x(\ph))\ ,
\end{equation}
where $j_x(\ph)=(x,\ph(x),\partial\ph(x),\dots)$ is the jet prolongation of $\ph$ and $f$ is a density-valued function on the jet bundle.
For these functionals, the derivatives are supported on the thin diagonal
\begin{equation}
D_n=\{(x_1,\dots,x_n)\in M^n, x_1=\dots=x_n\}
\end{equation}
and thus smooth for $n>1$ only when they vanish. 

The singularities of distributions can be analyzed using the concept of the wave front set. On Minkowski spacetime, this concept arises in the study of the decay properties of the Fourier transform of the given distribution multiplied by some test function. A pair of a spacetime point $x$ and a nonzero momentum $k$ is an element of the wave front set of a given distribution $t$ if, for any test function $f$ with $f(x)\neq 0$ and some open cone around $k$, the Fourier transform of $f t$ does not decay fast (i.e. faster than any power) inside the cone. The notion of the WF set can be generalized to an arbitrary smooth manifold $M$, and it is defined as a subset of $T^*M$.

As the first example we will consider the Dirac $\delta$ distribution. Since $\left<f\delta,e^{ik\bullet}\right>=f(0)$, where $e^{ik\bullet}(x)=e^{ikx}$, if follows that for any choice of  $f$ with $f(x)\neq 0$ the Fourier transform of $f t$ does not decay fast in any direction and hence the wave front set of $\delta$ is
\begin{equation}
\mathrm{WF}(\delta)=\{(0,k),k\neq 0\}\ .
\end{equation}
Another important example is the distribution $f\mapsto \lim_{\epsilon\downarrow0}\int \frac{f(x)}{x+i\epsilon} dx$. Note that its Fourier transform is
\be
\lim_{\epsilon\downarrow0}\int  \frac{f(x)}{x+i\epsilon}e^{ikx} dx=-i\int_k^{\infty}\hat{f}(k') dk'\ .
\ee
and $\int_k^{\infty}\hat{f}(k')dk'$ decays strongly as $k\to\infty$, while for $k\to -\infty$ we obtain
\be
\lim_{k\to-\infty}\int_k^{\infty}\hat{f}(k')dk'=2\pi f(0)\,.
\ee  
We can now conclude that
\be
\mathrm{WF}(\lim_{\epsilon\downarrow0}(x+i\epsilon)^{-1})=\{(0,k),k<0\}\ .
\ee

For more information on wave front sets see  \cite{Hoer} or chapter 4 of  \cite{BaeF}. Using WF sets we can formulate a sufficient condition for a pointwise product of distributions to be well defined. Let $t$ and $s$ be distributions on $M$. The Whitney sum (i.e pointwise sum) of their wave front sets is defined by
\be
\mathrm{WF}(t)+\mathrm{WF}(s)=\{(x,k+k')|(x,k)\in\mathrm{WF}(t),(x,k')\in\mathrm{WF}(s)\}
\ee
If this set does not intersect the zero section of $T^*M$, then we can define the pointwise product $ts$ as
\be
\langle ts,fg\rangle=\frac{1}{(2\pi)^n}\int\, \widehat{tf}(k)\widehat{sg}(-k)dk\,,
\ee 
where $f,g\in\Dcal$ are chosen with sufficiently small support. To see that the integral above converges, note that if $k\not=0$, then either $\widehat{tf}$ is fast-decaying in a conical neighborhood around $k$ or 
$\widehat{sg}$  is fast-decaying in a conical neighborhood around $-k$, while the other factor is polynomially bounded.

Beside the criterion for multiplying distributions, WF sets provide also a characterization of the propagation of singularities. Let $P$ be a partial differential operator and $\sigma_P$ its principal symbol. We can interpret $\sigma_P$ as a function on the cotangent bundle $T^*M$, which carries a structure of a symplectic manifold. With the use of the canonical symplectic form, 1-forms on $T^*M$ can be canonically identified with vector fields. Let $X_P$ be the vector field (called  the Hamiltonian vector field) corresponding to the 1-form $d\sigma_P$. In coordinates it is given by
\[
X_P=\sum_{i=1}^n\frac{\partial \sigma_P}{\partial {k}_j}\frac{\partial}{\partial x_j}-\frac{\partial \sigma_P}{\partial x_j}\frac{\partial}{\partial {k}_j} \ .
\]
Let $(x_j(t),{k}_j(t))$ be a curve that fulfills the system of equations (Hamilton's equations):
\begin{align*}
\frac{dx_j}{dt}&=\frac{\partial \sigma_P}{\partial {k}_j}\,,\\
\frac{d{k}_j}{dt}&=-\frac{\partial \sigma_P}{\partial x_j}\,.
\end{align*}
We call a solution $(x_j(t),{k}_j(t))$ of the above equations an \textit{integral curve} of $X_P$ and the bicharacteristic flow is defined as the set of all such solutions. Along this flow
 $\frac{d\sigma_P}{dt}=X_P(\sigma_P)=0$, so $\sigma_P$ is conserved under the bicharacteristic flow. We are now ready to state \textit{the theorem on the propagation of singularities}: the  wave front set of a solution $u$ of the equation $Pu=f$ with $f$ smooth is a union of orbits of the Hamiltonian flow $X_P$ on the characteristics $\cha P=\{(x,k)\in T^*M|\sigma(P)(x,k)=0\}$ of $P$. 

For hyperbolic differential operators on globally  hyperbolic spacetimes (for example $\MM$), the set of characteristics is the light cone, and the principal symbol is the metric on the cotangent bundle. For such operators the wave front set of solutions is therefore a union of null geodesics $\gamma$  together with their cotangent vectors $k=g(\dot{\gamma},\cdot)$.
\section{The Peierls bracket and the $\star$-product}
As outlined in the Introduction, we start our construction of a pAQFT model from the classical theory. To this end, we equip the space of functionals on the configuration space with a Poisson structure provided with the so called Peierls bracket. This bracket, introduced in \cite{Pei}, is the off-shell extension of the canonical bracket of classical mechanics, which is defined only on the space $\E_S$ of solutions to the equations of motion. To see how this works, we will start in a setting which resembles closely classical mechanics and then show the relation with the Peierls method on a concrete example. 
\subsection{Canonical formalism and the approach of Peierls}\label{canandP}
Let us start with the free scalar field with the field equation 
\be\label{KGeq}
P\ph=0\,,
\ee 
where $P=\Box+m^2$ is the Klein-Gordon operator. For this equation the retarded and advanced Green's functions exist. We also know that
for every $f\in\mathcal{E}$ whose support is past and future compact, $\Delta f$ is a solution to \eqref{KGeq}. 
Conversely, every smooth solution of the Klein Gordon equation 
is of the form $\Delta f$ for some $f\in \mathcal{E}(M)$ with future and past compact support. 

Without loss of generality, the spacetime can be assumed to be of the form $M=\RR\times\Sigma$ with Cauchy surfaces $\{t\}\times\Sigma$, $t\in\RR$.
The space of Cauchy data $\Sigma\ni\vr{x}\mapsto(\ph(t,\vr{x}),\dot{\ph}(t,\vr{x}))$ on the surface $\{t\}\times\Sigma$ is
\[
\mathcal{C}=\{ (\phi, \psi) \in \Ecal(\Sigma) \times \Ecal(\Sigma) \}\,,
\]
where $\Ecal(\Sigma)\doteq\Ci(\Sigma,\RR)$.
This space is isomorphic to $\E_{S}$, the of smooth solutions to \eqref{KGeq}.

%
As in classical mechanics, equations of motion can be derived from the least action principle. Elements of $\mathcal{C}$ play the role of generalized coordinates and generalized velocities, while a smooth trajectory $t\mapsto \phi(t)$, $t\in\RR$ is a function which assigns to an instant of time $t$ a function 
 $\phi(t)\in \Ecal(\Sigma)$ such that trajectories $\phi$ are in one to one correspondence with field configurations $\ph:(t,\vr{x})\to\phi(t)(\vr{x})$, i.e. elements of $\E$.

The Lagrangian $L$ associates to every compact region $K\subset\Sigma$ a functional $L_K$ on $\mathcal{C}$, typically given in terms of a Lagrangian density $\Lcal$,
\[L_K(\phi,\psi)=\int_K\Lcal(\phi(\vr{x}),\nabla\phi(\vr{x}),\psi(\vr{x}))d\sigma(\vr{x})\ ,\] 
 and the action is, for every compact $K\subset \Sigma$ and every finite time interval $I$, a function on the space of trajectories defined by
\be\label{action}
S_{I\times K}(\phi)=\int_I L_K(\phi(t),\dot{\phi}(t))dt=
\int_I\left(
\int_K \mathcal{L}(\ph(t,\vr{x}),\nabla_{\vr{x}}\ph(t,\vr{x}),\dot{\ph}(t,\vr{x}))d\sigma_t(\vr{x})\right)dt\ .
\ee
Solutions are configurations for which, for all compact $K$ and $I$, $S_{I\times K}$ is stationary under variations $\delta \phi$ with support in the interior of $I\times K$.
If e.g. $\Lcal$ is the Lagrangian density of the free scalar field, then the least action principle yields \eqref{KGeq} as the equation of motion. 

Now let $F,G$ be two functions on the space of trajectories which depend only on the restriction of the trajectory to  $[t_1,t_2]\times K$ for some compact $K\subset\Sigma$ and $t_1<t_2$.
Let $\E_S$ be the space of solutions for an action $S$, and let $r_{\lambda G}:\E_S\to\E_{S+\lambda G}$ be the map which associates to a solution for $S$ a solution for $S+\lambda G$ such that both solutions coincide for $t<t_1$ ($r_{\lambda G}$ is called the retarded M\o ller map).
Following the idea of Peierls, we consider the change of $F$ under the change of the action and set, for a solution $\phi\in\E_S$, 
\[ \mathrm{D}_GF(\phi)=\frac{d}{d\lambda}|_{\lambda=0}F(r_{\lambda G}(\phi))\ .\]
 Similarly, we introduce the advanced M\o ller map $a_{\lambda F}:\E_S\to\E_{S+\lambda F}$ where the solutions coincide for $t>t_2$, and set
  \[\reflectbox{D}_FG(\phi)=\frac{d}{d\lambda}|_{\lambda=0}G(a_{\lambda F}(\phi))\ .\]
  The Peierls bracket of $G$ and $F$ is now defined by
\be\label{Peierls1}
\{G,F\}_{\textrm{Pei}}\doteq  \mathrm{D}_GF-\reflectbox{D}_FG\,.
\ee
The advantage of the Peierls bracket is the fact that it is defined covariantly, directly in the Lagrangian formalism. As it stands, the Peierls bracket of two functionals is only defined on solutions, and one has to prove that it depends only on the restriction of the functionals to the space of solutions. In order to show that it satisfies the Jacobi identity, one has to extend it to a neighborhood of the space of solutions. It is, however, possible and also convenient to extend it to a Poisson bracket on functions of arbitrary configurations $\E$ (not only of solutions $\E_S$). We will now derive another formula for the Peierls bracket (formula \eqref{Peierls2} from the Introduction), which makes use of retarded and advanced Green's functions for normally hyperbolic operators. Next we will show that, restricted to the solution space, \eqref{Peierls2} is equivalent to the canonical bracket. 
\subsection{The generalized Lagrangian formalism}
Before we continue, there is one small modification to the classical Lagrangian formalism, which we have to perform in order to make the quantization simpler. In formula \eqref{action}, we have smeared the Lagrangian density $\Lcal(\vr{x},t)$ with a characteristic function of a certain compact region. Such sharp cut-offs would introduce additional divergences in the quantum theory, which we wish to avoid. Therefore, we replace the characteristic function by a smooth function that is equal to 1 on a sufficiently large region. Actually, it is convenient to consider \textit{all} possible cutoffs and define the \textit{generalized Lagrangian} as a map $L$ from $\Dcal$ to the space $\Fcal_{\textrm{loc}}$ of local functionals on $\E$. We require that 
\[
L(f+g+h)=L(f+g)-L(g)+L(g+h)\,,
\]
for $f,g,h\in\Dcal$ and $\supp\,f\cap\supp\,h=\varnothing$. We also want
\[
\supp(L(f))\subseteq \supp(f)\,,
\]
where the support of a smooth functional $F\in\Ci(\E,\CC)$ is defined as
\begin{align}\label{support}
\supp F\doteq\{ & x\in M|\forall \text{ neighborhoods }U\text{ of }x\ \exists \ph,\psi\in\Ecal, \supp\,\psi\subset U\,,
\\ & \text{ such that }F(\ph+\psi)\not= F(\ph)\}\ .\nonumber
\end{align}
The action is an equivalence class of Lagrangians, where $L_1\sim L_2$  if
\be\label{equ}
\supp (L_{1}-L_{2})(f)\subset\supp\, df\,.
\ee
The Euler-Lagrange derivative is a map $S':\E\to\Dcal'$ defined as
\be\label{ELd}
\left<S'(\ph),h\right>=\left<L(f)^{(1)}(\ph),h\right>\,,
 \ee
 with $f\equiv 1$ on $\supp h$. Note that $S'\in\Gamma(T^*\Ecal)$. The field equation is now the condition that
\be
 S'(\ph)=0\label{eom}\,,
\ee
which coincides with the condition obtained from the variation of \eqref{action}. We model observables as multilocal functionals on $\E$ (i.e. products of local functionals). The maps $F$, $G$ considered in the previous section are examples of such functionals. The space of multilocal functionals on the space of solutions to \eqref{eom} is given by the quotient $\Fcal/\Fcal_0$, where $\Fcal_0$ denotes the space of multilocal functionals that vanish on $\E_S$. 

The second variational derivative of the action is defined by
\[
\left<S''(\ph),h_1\otimes h_2\right>\doteq \left<L^{(2)}(f)(\ph),h_1\otimes h_2\right>\,,
\]
where $f\equiv 1$ on $\supp\,h_1$ and $\supp\,h_2$. $S''$ defined in such a way is a symmetric two tensor on the affine manifold $\E$ (equipped with the smooth structure induced by $\tau_W$) and for each $\ph$ it induces an operator from $\Dcal$ to $\Dcal'$. Moreover, since $L(f)$ is local, the second derivative has support on the diagonal, so $S''(\ph)$ 
can be evaluated on smooth functions $h_1$, $h_2$, where only one of them is required to be compactly supported. This way we obtain an operator (the so called linearized Euler-Lagrange operator) $E'[S](\ph):\E\rightarrow \Dcal'$.

We want to show now that the original formula of Peierls \eqref{Peierls1} is equivalent to \eqref{Peierls2}, if $E'[S](\ph)$ is a normally hyperbolic operator.  Let $G\in\mathcal{F}_\loc$ be a local functional. We are interested in the flow $(\Phi_{\lambda})$ on $\E$ which deforms solutions of the original field equation $S'(\ph)=0$ to those of the perturbed equation $S'(\ph)+\lambda G^{(1)}(\ph)=0$. Let $\Phi_0(\ph)=\ph$ and 
\be\label{flow}
\frac{d}{d\lambda}\left.\left(S_\Mcal'(\Phi_{\lambda}(\ph))+G^{(1)}(\Phi_{\lambda}(\ph))\right)\right|_{\lambda=0}=0 \ .
\ee
The vector field $\ph\mapsto X(\ph)=\frac{d}{d\lambda}\Phi_{\lambda}(\ph)|_{\lambda=0}$ satisfies the equation
\be\label{X:eq}
\left<E'[S](\ph),X(\ph)\right>+G^{(1)}(\ph)=0\,.
\ee
Let $\Delta^{R/A}_S(\ph)$ be the retarded/advanced Green's function of the normally hyperbolic operator $E'[S](\ph)$ and let $\Delta_S(\ph)=\Delta^R_S(\ph)-\Delta^A_S(\ph)$ be the causal propagator. We obtain now two distinguished solutions to the equation \eqref{X:eq},
\be
X^{R/A}(\ph)=\left<\Delta_S^{R/A}(\ph),G^{(1)}(\ph)\right>\ .
\ee
Note that $X^R(\ph)=(\mathrm{D}_G\Phi)(\ph)$, where $\Phi$ is the evaluation functional $\Phi_x(\ph)\doteq \ph(x)$. The difference $X=X^R-X^A$ defines a vector field $X\in\Gamma(T\E(M))$ and it follows that
\[
\{G,F\}_{\textrm{Pei}}(\ph)\doteq  \mathrm{D}_GF(\ph)-\reflectbox{D}_FG(\ph)=\left<F^{(1)}(\ph),\Delta_S^{R/A}(\ph)G^{(1)}(\ph)\right>\,.
\]

Now we prove the equivalence between \eqref{Peierls2} and the canonical bracket. We fix a Cauchy surface $\{t\}\times\Sigma$. Note that, given Cauchy data $(\phi,\psi)\in\Ccal$, we can write the unique solution $\ph$ corresponding to these Cauchy data as
\be\label{Cauchy}
\ph(x)=\beta(\phi,\psi)(x)\equiv\int_{\Sigma}\left(\Delta_S(x;t,\vr{y})\psi(\vr{y})-\frac{\partial}{\partial t}{\Delta}_S(x;t,\vr{y})\phi(\vr{y})\right)d\sigma_t(\vr{y})\,.
\ee
Canonical momenta are obtained as distributional densities by
\[
\langle\pi(\phi,\psi),h\rangle\doteq \frac{d}{d\lambda}|_{\lambda=0}L_K(\phi,\psi+\lambda h)\, ,\ h\in\Dcal(K)\ .
\]
We assume that for the Lagrangians of interest $\pi$ is always smooth.
The phase space is then
\be\label{phase}
\Pcal=\Ecal(\Sigma)\times \Ecal_d(\Sigma)\,,
\ee
where $\Ecal_d(\Sigma)$ is the space of smooth densities. The tangent space $T_{(\phi,\psi)}\Pcal$ of $\Pcal$ at some point $(\phi,\psi)$ consists of the compactly supported elements $(f,g)\in\Pcal$. 
The phase space has the canonical symplectic form
\[
\sigma_{(\phi,\psi)}((f_1, f_2), (g_1, g_2)) = \int_\Sigma (f_1 g_2 - f_2 g_1).
\]
Note that  $\Ecal(\Sigma) \times \Ecal_d(\Sigma)\subset \Ecal(\Sigma) \times \Dcal'(\Sigma)\cong T^*(\Ecal(\Sigma) )$, so $(\Pcal,\sigma)$ is indeed the analog of the phase space in classical mechanics. 

For simplicity we consider an action $S$ induced by a Lagrangian $L$ which depends on $\dot{\phi}$ only through the kinetic term $\frac{1}{2}\dot{\phi}^2$, hence $\pi(\vr{y})\doteq\dot{\phi}(\vr{y})d\sigma_t(\vr{y})$. Let $\al:(\phi,\pi)\mapsto(\phi,\dot{\phi})$ and $\tilde{\beta}\doteq \beta\circ\al:\Pcal\rightarrow \Ecal_S$. We can now prove the equivalence of the canonical and the Peierls bracket. Let $F,G\in\Fcal$. Using \eqref{Cauchy} we obtain
\begin{multline*}
\{F\circ\tilde{\beta},G\circ\tilde{\bet}\}_{\textrm{can}}=\int_{\Sigma}\left(\left\langle\tfrac{\delta F}{\delta\ph}\circ\tilde{\bet},\tfrac{\delta \tilde{\bet}}{\delta \phi(\vr{x})}\right\rangle\left\langle\tfrac{\delta G}{\delta\ph}\circ\tilde{\bet},\tfrac{\delta \tilde{\bet}}{\delta \pi(\vr{x})}\right\rangle-\left\langle\tfrac{\delta F}{\delta\ph}\circ\tilde{\bet},\tfrac{\delta \tilde{\beta}}{\delta \pi(\vr{x})}\right\rangle\left\langle\tfrac{\delta G}{\delta\ph}\circ\tilde{\bet},\tfrac{\delta \tilde{\bet}}{\delta \phi(\vr{x})}\right\rangle\right)=\\
\left<\Theta,F^{(1)}\circ\tilde{\bet}\otimes G^{(1)}\circ\tilde{\bet}\right>\,,
\end{multline*}
where $\Theta$ is given by
\[
\Theta(z',z)=\int_{\Sigma}\left(\dot{\Delta_S}(z';t,\vr{x}){\Delta_S}(z;t,\vr{x})-\dot{\Delta_S}(z;t,\vr{x}){\Delta_S}(z';t,\vr{x})\right) d\sigma(\vr{x})\,.
\]
From general properties of the causal propagator $\Delta_S$ (the generalization of \eqref{Cauchy} to distributional Cauchy data) it follows that the convolution $\Theta$ above is equal to $\Delta_S$. Hence, on the solution space $\Ecal_S$, 
\[
\{F\circ\tilde{\beta},G\circ\tilde{\beta}\}_{\textrm{can}}=\{F,G\}_{\textrm{Pei}}\circ\tilde{\beta}\,.
\]
\subsection{Example: the Poisson bracket of the $\ph^4$ interaction}
In this section, following \cite{FR15}, we give another argument for the equivalence of the Peierls and the canonical bracket on the example of the $\ph^4$ interaction. Consider the generalized Lagrangian
\[
L(f)(\ph)=\int\limits_{\MM} \left(\frac{1}{2}\nabla_\mu\ph\nabla^\mu\ph-\frac{m^2}{2}\ph^2-\frac{\la}{4!}\ph^4\right)fd\mu\,,
\]
where $d\mu(x)$ is the invariant measure $\mu$ on $M$, induced by the metric. Then $S'(\ph)=-\left((\square+m^2)\ph+\frac{\lambda}{3!}\ph^3\right)$ and $E'[S](\ph)$ is the linear operator
\be
-\left(\square +m^2+\frac{\lambda}{2}\ph^2\right)
\ee
(the last term is to be understood as a multiplication operator). The Peierls bracket is
\be
\Poi{\Phi_x}{\Phi_y}_{\textrm{Pei}}=\Delta_S(\Phi)(x,y)\,,
\ee
where $\Phi_x$, $\Phi_y$ are evaluation functionals on $\E$ 
and   $x\mapsto \Delta_S(\ph)(x,y)$ is a solution (at $\ph$) of the  linearized equation of motion with the initial conditions
\be
\Delta_S(\ph)(y^0,\mathbf{x};y^0,\mathbf{y})=0\ ,\ \frac{\pa}{\pa x^0}\Delta_S(\ph)(y^0,\mathbf{x};y)=\delta(\mathbf{x},\mathbf{y})\ .
\ee
This coincides with the Poisson bracket in the canonical formalism. Namely, let $\ph\in\E_S$,
then
\be
0=\left\{(\square +m^2)\Phi_x+\frac{\lambda}{3!}\Phi_x^3,\Phi_y\right\}_
{\textrm{can}}=\left(\square +m^2+\frac{\lambda}{2}\Phi_x^2\right)\left\{\Phi_x,\Phi_y\right\}_{\textrm{can}}\,.
\ee
In the first step we used the fact that $\ph$ is a solution of the equations of motion and in the second step we used the fact that the canonical bracket is a derivation in both arguments. We can see from the equation above that the canonical Poisson bracket satisfies the linearized field equation with the same initial conditions as the Peierls bracket. The uniqueness of solutions to these linearized equations implies that in fact $\{\Phi_x,\Phi_y\}_{\textrm{can}}=\Poi{\Phi_x}{\Phi_y}_{\textrm{Pei}}$, on $\Ecal_S$. This clearly extends to general functionals since $\Poi{F}{G}_{\textrm{Pei}}=\left<\Poi{\Phi_x}{\Phi_y}_{\textrm{Pei}},F^{(1)}\otimes G^{(1)}\right>$, and similarly for $\{.,.\}_{\textrm{can}}$.
\subsection{Geometrical structures in classical theory}\label{geom}
The Peierls bracket, from now on denoted by $\{.,.\}_S$, introduces a symplectic structure on the space $\Fcal/\Fcal_0$ of on-shell multilocal functionals on $\Ecal_S$. We can find a nice geometrical interpretation for this space using some basic notions of symplectic geometry. Let us assume that $S$ is quadratic, so the equations of motion are of the form $S'(\ph)=P\ph=0$ for some normally hyperbolic differential operator $P$. The  space of solutions $\Ecal_S$ is a vector space and hence an infinite dimensional manifold\footnote{There is another natural way to introduce a smooth manifold structure on $\Ecal_S$. We define the atlas where charts are given by maps
$\ph+\cD\to \mathcal{E}_{S,sc}$, $\ph+\vec{\ph}\mapsto \Delta_S\vec{\ph}$, with $\ph\in\cE_S$, where $\mathcal{E}_{S,sc}$ is the space of solutions with compactly supported Cauchy data. We have $\Delta_S:\Dcal\rightarrow \mathcal{E}_{S,sc}$ and we equip 
 $\mathcal{E}_{S,sc}$ with the final topology with respect to all curves of the form $\la\mapsto \ph+ \Delta_S(\vec{\ph}(\la))$, where $\la\mapsto \vec{\ph}(\la)$ is a smooth curve in $\Dcal$. This gives $\Ecal_{S}$ the structure of an affine manifold in the sense of convenient calculus \cite{Michor}} with a tangent space $T\Ecal_S=\Ecal_S\times \Ecal_S$ and cotangent space  $T^*\Ecal_S=\Ecal_S\times \Ecal_S'$. If $F,G$ are multilocal functionals on $\Ecal_S$, then their first functional derivatives are smooth, so 
 \[
 F^{(1)}(\ph),G^{(1)}(\ph)\in \Ecal_S'\cap(\Dcal/\{u\in\Dcal|\left<u,\ph\right>=0\,,\forall\ph\in\Ecal_S \})\]
 for all $\ph\in\Ecal_S$, i.e. $\ph\in\ker P$. Since $P$ is a normally hyperbolic operator, one can show (see for example \cite{FR15} for the proof in a more general setting) that  $\{u\in\Dcal|\left<u,\ph\right>=0\,,\forall\ph\in\Ecal_S \}\cong P\Dcal$, so functional derivatives of multilocal functionals are one forms in  $\Gamma(\Ecal_S\times \Dcal/P\Dcal)\subset \Gamma(T^*\Ecal_S)$.
 
The causal propagator $\Delta_S$ induces a Poisson structure on $\Fcal$, which is also well defined on the quotient $\Fcal/\Fcal_0$, as $\Fcal_0$ is a Poisson ideal with respect to this structure. We can also use $\Delta_S$ to map one-forms in $\Gamma(\Ecal_S\times \Dcal/P\Dcal)\subset \Gamma(T^*\Ecal_S)$ to one-vectors in $\Gamma(\Ecal_S\times\Ecal_{S,\textrm{sc}})\subset\Gamma(T\Ecal_S)$, where ``$\textrm{sc}$'' indicates spacelike-compact support. To see how it works, note that $\Delta_S$ induces an operator $\Dcal\rightarrow \Ecal$ and $\ker\Delta_S=P \Dcal$, so $\Delta_S$ is well defined on equivalence classes in $\Dcal/P\Dcal$. To show that $\Delta_S$ is invertible on this space, it remains to show that it is surjective. We recall here the standard argument, which can also be found in \cite{FR15}.  Let $f$ be a solution with a spacelike-compact support, $\chi \in \Ecal$, and $\Sigma_1, \Sigma_2$ be Cauchy surfaces such that $\Sigma_1 \cap J_+(\Sigma_2)=\varnothing$. Assume $\chi(x)=0$ for $x \in J_-(\Sigma_1)$ and $\chi(x)=1$ for $x \in J_+(\Sigma_2)$. Then $P \chi f=0$ outside of the time slice bounded by $\Sigma_1$ and $ \Sigma_2$ ($\chi =$const. there) which implies that $P\chi f$ has compact support. Hence,
\[
 \Delta_S P \chi f =  \Delta_S^R P \chi f +  \Delta_S^A P(1-\chi)f = f.
\]

We can now assign to a form $F^{(1)}$, the vector $\left<\Delta_S F^{(1)},.\right>$. On $\Ecal_{S,\textrm{sc}}$ we have the natural symplectic structure $\sigma_1$:
\[
\sigma_1(f,g)=\int_\Sigma(f\wedge *dg-*df\wedge g) = \int_\Sigma  (f (\partial_n g) - (\partial_n f)g)dvol_\Sigma,
\]
where $\partial_n$ is the normal derivative on $\Sigma$ ($\partial_n f = n^{\mu} \partial_{\mu}f$, $n^{\mu} \xi_{\mu}=0 $ for $\xi \in T\Sigma$, $n^{\mu}n_{\mu}=1$.). Obviously, $\sigma_1$ extends to a constant 2-form on $\Ecal_S$. The relation between the 2-form $\sigma_1$ and the ``bi-vector field''\footnote{Since $\Delta_S$ is a bi-distribution rather than a smooth function, the map  $\ph\mapsto \Delta_S$ doesn't induce an actual bi-vector field on $\Ecal$, but belongs to a suitable completion of $\Gamma(\La^2T\Ecal)$.} $\Delta_S$ is given by (see for example \cite{Weise} for a proof based on the ideas of \cite{CrWitten})
\[
\sigma_1(\Delta_S F^{(1)},\xi)=\left< F^{(1)},\xi\right>\,,
\]
where $\xi\in\Gamma(\Ecal_S\times\Ecal_{S,\textrm{sc}})$. In this sense we can think of $\Delta_S$ as the ``inverse'' of the symplectic structure $\sigma_1$. Setting $\xi=\Delta_S G^{(1)}$ for $G\in\Fcal$, we obtain
\[
\sigma_1(\Delta_S F^{(1)},\Delta_S G^{(1)})=\left< F^{(1)},\Delta_S G^{(1)}\right>\,.
\]

If $S$ is not quadratic, the situation is more complicated, since $S'$ induces non-linear equations of motion. It turns out that for many classes of physically interesting systems solutions of $S'(\ph)=0$ develop singularities after a finite time, despite  starting from smooth Cauchy data. Therefore the space of globally smooth solutions $\Ecal_S$ might be very small and it does not necessarily capture all the interesting features of the theory. Moreover, in general it is not clear if $\Ecal_S$ can be equipped with a manifold structure in the sense of infinite dimensional differential geometry (see for example \cite{AMM} for the results on the space of solutions of Einstein's equations). A more general structure like a stratified space might be necessary.

For non-linear equations of motion it is therefore more convenient to replace the space of functionals on the space of solutions with the quotient $\Fcal_S:=\Fcal/\{\left<S',X\right>,X\in \Vcal\}$, where $\Vcal\subset\Gamma(T\Ecal)$ is the space of vector fields that are derivations of $\Fcal$ and the duality denoted by $\left<.,.\right>$ is the contraction of a 1-form  $S'\in\Gamma(T^*\E(\Mcal))$ with a vector field $X$. We say that we take the quotient of $\Fcal$ by the \textit{ideal generated by the equations of motion}. If the equations of motion are linear and normally hyperbolic, this ideal coincides with $\Fcal_0$, so $\Fcal_S$ is exactly the space of multilocal functionals on $\Ecal_S$. In general, our point of view is more in line with the quantum theory and it avoids complications related to characterization of the geometrical structure of $\Ecal_S$. It is also close in spirit to the way one studies varieties in algebraic geometry.
\subsection{Deformation quantization}\label{Def:quant}
Deformation quantization is a method to construct quantum theories from the classical ones by deforming the commutative product on the space $\Fcal$ of functionals to a non-commutative product $\star$ on $\Fcal[[\hbar]]$ (the space of formal power series in $\hbar$. 
\be
F\star G=\sum\limits_{n=0}^\infty\hbar^n B_n(F,G)\,,
\ee
and we require that
\begin{align*}
B_0(F,G)&=F\cdot G\,,\\
B_1(F,G)-B_1(G,F)&=i\hbar\{F,G\}\,,
\end{align*}
where $(F\cdot G)(\ph)=F(\ph)G(\ph)$ is the pointwise product of functionals, $\{.,.\}$ is the Peierls bracket and the second condition is a realization of the idea that in the quantum theory one ``replaces canonical brackets with commutators''. The existence of higher order terms is necessary to avoid the Groenewald-van Hove no-go theorem. This result, established first for finite dimensional phase spaces, states that a Dirac type quantization prescription is not possible in the strict sense \cite{Groenewald,vanHove}. More concretely (see \cite{Waldmann}), consider the Lie algebra $\mathfrak{h}$ spanned by the canonical coordinate and momenta functions $q^1,\ldots,q^N,p_1,\ldots,p_N$ and $1$, equipped with the canonical Poisson bracket $\{.,.\}_{\mathrm{can}}$. This is a Lie subalgebra of $\mathfrak{g}\doteq(\mathrm{Pol}(T^*\RR^N),\{.,.\}_{\mathrm{can}})$ (polynomials on the phase space). The Groenewald-van Hove Theorem states that
there exists no faithful irreducible representation of $\mathfrak{h}$ by operators on a dense domain of some Hilbert space which can be extended to a representation of $\mathfrak{g}$, so there is no quantization map $Q$ from $\frakg$ to the space of operators on some Hilbert space $\Hcal$, such that
\[
[Q(f),Q(g)]=i\hbar Q(\{f,g\})\,.
\]
Deformation quantization \cite{BFFLS,BFFLS2} provides a way out since it weakens the above condition to 
\[
[Q(f),Q(g)]=Q([f,g]_{\star})=i\hbar Q(\{f,g\})+\Ocal(\hbar^2)\,.
\]

In field theory, as we have seen in the previous section, the space of functions on the $N$-dimensional phase space is replaced by $\Fcal$, the space of multilocal functionals on $\Ecal$, which is now infinite dimensional. The Poisson structure is provided by the Peierls bracket, defined with the use of the causal propagator $\Delta_S$. In the simplest case, when $S$ is quadratic, one can construct the $\star$-product using a Moyal-type formula. To avoid the functional analytic problems, for the moment we consider only regular functionals $\Fcal_{\reg}$, i.e. those for which $F^{(n)}(\ph)$ is a smooth compactly supported section for all $n\in\NN$, $\ph\in\Ecal$. On $\Fcal_{\reg}[[\hbar]]$ we can now define
\be\label{starprod}
(F\star G)(\ph)\doteq\sum\limits_{n=0}^\infty \frac{\hbar^n}{n!}\left<F^{(n)}(\ph),\left(\tfrac{i}{2}\Delta_S\right)^{\otimes n}G^{(n)}(\ph)\right>\,,
\ee 
which can be formally written as $e^{\frac{i\hbar}{2}\left\langle\Delta_S,\frac{\delta^2}{\delta\ph\delta\ph'}\right\rangle}F(\ph)G(\ph')|_{\ph'=\ph}$. 

Let us consider the example of a free scalar field and regular functionals of the form 
\[
F_f(\ph)=\int_{\MM}f(x)\ph(x) d^4x\equiv \int f\ph d\mu \,,\qquad \textrm{where}\ f\in \Dcal\,.
\]
We can now define  $\Wcal(f)\doteq\exp(iF_f)$ and check that
\begin{align*}
\left< (\Wcal(f))^{(1)}(\ph),h\right>&= \frac{d}{d\lambda}\left( \Wcal(f) (\ph + \lambda h)\right)|_{\lambda=0}= \frac{d}{d \lambda}  e^{i \int  f (\ph + \lambda h)d\mu} \big|_{\lambda=0}=\\
&=\left( i \int f h\,d\mu_g\right) \Wcal(f)(\ph).\\
\end{align*}
and thus
\[
\left<(\Wcal(f))^{(n)}(\ph),h^{\otimes n}\right>=  \left( i \int f h\,d\mu_g\right)^n \Wcal(f)(\ph).
\]
Inserting this into the $\star$-product formula, we find,
\begin{eqnarray}
\Wcal(f) \star \Wcal(\tilde{f}) &=&  \sum_{n=0}^{\infty} \left( \frac{i \hbar}{2}\right)^n \frac{(-1)^n}{n!} \left( \int \Delta_S(x,y) \tilde{f}(y)f(x)d\mu_g(x)d\mu_g(y)\right)^n\Wcal(f+\tilde{f})\nonumber\\
&=& e^{-\frac{i \hbar}{2}\Delta_S(f,\tilde{f})} \Wcal(f+\tilde{f}),\label{Weyl}
\end{eqnarray}
which reproduces the Weyl relations. 

Having the interacting theory in mind, we will need to extend the star product to functionals more singular than the elements of $\Fcal_{\reg}$, including, in particular, the non-linear local functionals. To understand possible obstructions to this extension we have to analyze the singularity structure of $\Delta_S$. Using the theorem of propagation of singularities (see section \ref{sec:Functional derivatives}), we find that \cite{Rad}
\[
\WF(\Delta_S)=\{(x,k;x′,-k')\in \dot{T}^*M^2|(x,k)\sim(x',k')\}\,,
\]
where $\dot{T}$ denotes the tangent bundle minus the zero section and $(x,k)\sim(x',k')$ means that there exists a lighlike geodesic connecting $x$ and $x'$, to which $k$ is co-tangent and $k'$ is a parallel transport of $k$. 
We observe that the WF set of $\Delta_S$ is composed of two parts: one with $k\in (\overline{V}_+)_x$ and another with $k\in (\overline{V}_-)_x$, where $\overline{V}_\pm$ is (the dual of) the closed future/past lightcone  This observation allows one to decompose  $\Delta_S$ into two distributions with WF sets corresponding to these two components. Such a decomposition is a local version of the decomposition according to positive and negative energies  \cite{Rad}. 
%
%
Therefore we can split $\Delta_S$ into
\[
\tfrac{i}{2}\Delta_S=\Delta_S^+-H\,,
\]
where the WF set of $\Delta_S^+$ is
\begin{equation}\label{spectrum}
\WF(\Delta_S^+)=\{(x,k;x′,-k')\in \dot{T}^*M^2|(x,k)\sim(x',k'), k\in (\overline{V}_+)_x\}\,,
\end{equation}
and we also require that $\Delta_S=2\textrm{Im} (\Delta_S^+)$ and that $\Delta_S^+$ is a distributional bisolution to the field equation and is of positive type (i.e. $\left< \Delta_S^+,\bar{f}\otimes f\right>\geq 0$). 
On Minkowski space one could choose $\Delta_S^+$ as the  Wightman 2-point-function. On general globally hyperbolic spacetimes such a decomposition always exists but is not unique. If $H$ and $H'$ correspond to two such choices of decomposition, then $H-H'$ is a smooth symmetric bisolution to the field equations.

We can now replace $\tfrac{i}{2}\Delta_S$ with $\Delta_S^+$ in \eqref{starprod} and the new product, denoted by $\star_H$ can be extended from $\Fcal_{\reg}$ to $\Fcal_{\mc}$ defined as the space of functionals with functional derivatives satisfying
 \be\label{mlsc}
\WF(F^{(n)}(\ph))\subset \Xi_n,\quad\forall n\in\NN,\ \forall\ph\in\E\,,
\ee
where $\Xi_n$ is an open cone defined as 
\be\label{cone}
\Xi_n\doteq T^*\MM^n\setminus\{(x_1,\dots,x_n;k_1,\dots,k_n)| (k_1,\dots,k_n)\in (\overline{V}_+^n \cup \overline{V}_-^n)_{(x_1,\dots,x_n)}\}\,,
\ee
where $(\overline{V}_{\pm})_x$ is the closed future/past lightcone understood as a conic subset of $T^*_x\MM$.

On $\Fcal_{\reg}$ the two star products $\star$ and $\star_H$ are isomorphic structures and the intertwining map is given by 
\be
\al_H\doteq e^{\frac{\hbar}{2}\langle H,\frac{\delta^2}{\delta\ph^2}\rangle}\,,
\ee
so that
\be
F\star_HG=\al_H\left((\al^{-1}_HF)\star(\al_H^{-1}G)\right)\,,\qquad F,G\in\Fcal_{\reg}\,.
 \ee
 In the language of formal deformation quantization one says that products
 $\star$ and $\star_H$ are related by a gauge transformation, so they provide the same deformation quantization. In general a gauge transformation between star products is given by $F\mapsto F+\sum_{\hbar\geq 1}\hbar^n D_n(f)$, where each $D_n$ is a differential operator. In our case, $D_n=\frac{1}{n!}\left\langle \frac{1}{2}(H-H'),\frac{\delta^2}{\delta\ph^2}\right\rangle^n$.
   
Physically, the transition between  $\star$ and $\star_H$ corresponds to normal ordering, so introducing the $\star_H$-product is just an algebraic version of Wick's theorem. As stated before, the codomain of $\al_{H}:\Fcal_{\reg}\rightarrow\Fcal_{\reg}$ can be ``completed'' (with the use of the H\"ormander topology \cite{BDF09,Hoer}) to a larger space $\Fcal_{\mc}$ and we can also build a corresponding (sequential) completion $\al_H^{-1}(\Fcal_{\mc})$ of the domain. This amounts to extending  $\Fcal_\reg$ with all elements of the form $\lim_{n\rightarrow \infty}\al_H^{-1}(F_n)$, where $(F_n)$ is a convergent sequence in $\Fcal_{\mc}$. The quantum algebra $\fA$ of the free theory is defined as the space of families $F_H$,  labeled by possible choices of $H$, where $F_H\in\fA_H\doteq(\Fcal_{\mc}[[\hbar]],\star_H)$
 fulfill the relations
\[
 F_{H'} = \al_{H'-H} F_H\,,
\]
and the product is
\[
 (F \star G)_H = F_H \star_H G_H.
\]
We can summarize the relations between the algebraic structures we have introduced so far by means of the following diagram:
\[
\begin{CD}
(\Fcal_{\reg},\star)@>\al_H>>(\Fcal_{\reg},\star_H)\\
@V{\textrm{dense}}V{\cap}V@V{\textrm{dense}}V{\cap}V\\
\fA@<\al^{-1}_H<<(\Fcal_{\mc},\star_H)
\end{CD}
\]
A family of coherent states on $\fA$ is obtained by the prescription
\[
\omega_{H,\ph}(F)\doteq \al_H(F)(\ph)=F_H(\ph)\,,
\]
where $\ph\in\Ecal_S$. This makes sense since $F_H$ is a functional in $\Fcal_{\mc}$,  so evaluation at a field configuration $\ph$ is well defined. 

As an example we can consider the free scalar field with the generalized Lagrangian
\be\label{Lscalar}
L_0(f)(\ph)=\frac{1}{2}\int\limits_{\MM} (\partial_\mu\ph\partial^\mu\ph-m^2\ph^2)f\,d\mu\,.
\ee
Let us define  $\tilde\fA$ as the subalgebra of $\fA$ generated by the Weyl generators $\Wcal(f)\doteq\exp(iF_f)$. Since 
\[
\langle H,\frac{\delta^2}{\delta\ph^2}\rangle\left(i \int f\ph d\mu\right)^n=-\frac{n!}{(n-2)!}H(f,f)\left(i \int f\ph d\mu\right)^{n-2}\,,
\]
we conclude that
\[
\langle H,\frac{\delta^2}{\delta\ph^2}\rangle\left(\Wcal(f)\right)=-H(f,f)\Wcal(f)
\]
so
\[
\al_H\left(\Wcal(f)\right)=e^{-\frac{\hbar}{2}H(f,f)}\Wcal(f)\,.
\]
We can now consider a state obtained be evaluation at $\ph=0$. We see that
\[
\omega_{H,0}\left(\Wcal(f)\right)=e^{-\frac{\hbar}{2}H(f,f)}\,,
\]
so $H$ plays the role of the covariance of the state $\omega_{H,0}$.

Going on-shell corresponds to taking the quotient of  $\tilde\fA$  by the ideal  $\tilde\fA_0$  generated by the elements
\begin{equation}\label{ideal}
\Wcal((\square +m^2)f)-1\ , \ f\in\cD\ .
\end{equation}
Note that $S_0'(\ph)=(\square +m^2)\ph$, and using partial integration, we can conclude that $F_{(\square +m^2)f}(\ph)=\int S_0'(\ph) fd\mu=\left<S'_0,f\right>$, so taking the quotient by $\tilde\fA_0$ implements the free field dynamics. We denote $\tilde{\fA}/\tilde{\fA}_0$ by $\tilde{\fA}_{S_0}$, and we see that $\omega_{H,0}$ is well defined on $\tilde{\fA}_{S_0}$, as $H$ is a bisolution for the operator $P=\square +m^2$.
\subsection{Interpretation in terms of K\"ahler geometry}
There is an elegant geometrical interpretation of the structures introduced in the  previous section. Analogous to  K\"ahler geometry, $H$ plays the role of the Riemannian metric on $\mathcal{Y}\equiv\Dcal/P\Dcal$ and the 2-point function $\Delta_S^+=\frac{i}{2}\Delta_S+H$ is a Hermitian 2-form on $\mathcal{Y}$.

The pair $(H,\Delta_S)$ induces an anti-involution $J$ on $\Ycal$ (i.e. $J^2=-1$) and if $\Delta_S^+$ is a 2-point function of a quasi-free pure Hadamard state, the triple $(H,\Delta_S,J)$ is a  K\"ahler structure on $\mathcal{Y}$. To see how this come about, let us recall some well known results (see for example \cite{DG,AS} for proofs). Let $\mathcal{Y}^{\CC}$ denote the complexification of $\Ycal$. If $\Delta_S^{\CC}$ and $H^{\CC}$ are canonical extensions of $\Delta_S$ and $H$ to   $\mathcal{Y}^{\CC}$, then the following are equivalent:
\begin{enumerate}
\item $H^{\CC}+\frac{i}{2}\Delta_S^{\CC}\geq 0$ on $\mathcal{Y}^{\CC}$,
\item $|\left<f_1,\Delta_Sf_2\right>|\leq 2\left<f_1,Hf_1\right>^{1/2}\left<f_2,Hf_2\right>^{1/2}$, $f_1,f_2\in\Ycal$.
\end{enumerate}
We can complete $\Ycal$ with the product $(.,.)_H\doteq \left<.,H.\right>$ to a real Hilbert space $\Hcal$ and the inequality 2 implies that $\Delta_S$ is a bilinear form on $\Hcal$ with norm less or equal 2. Therefore, there exists an operator $A\in \Bcal(\Hcal)$ with $||A||\leq 1$ such that
\[
\left<f_1,\Delta_Sf_2\right>=2(f_1,Af_2)_H
\]
If $A$ has a trivial kernel, then we can just construct the polar decomposition $A=-J|A|$ and $J$ satisfies $J^2=-1$, so we can use it to equip $\Ycal$ with an almost-complex structure. More generally, following the proof of theorem 17.12 of \cite{DG}, we can define $\Ycal_{\textrm{sg}}\doteq \ker A$ and  $\Ycal_{\textrm{reg}}\doteq\Ycal_{\textrm{sg}}^{\perp}$. We set $A_{\mathrm{reg}}\doteq A\upharpoonright_{\Ycal_{\textrm{reg}}}$ and construct the polar decomposition $A_{\mathrm{reg}}=-J_{\mathrm{reg}}|A_{\mathrm{reg}}|$. If the dimension of $\Ycal_{\textrm{sg}}$ is even or infinite (which is the case in the situation we are interested in), then there exist an orthogonal anti-involution $J_{\mathrm{sg}}$ on $\Ycal_{\textrm{sg}}$ and we set $J=J_{\mathrm{reg}}\oplus J_{\mathrm{sg}}$.

 Note that $J$ induces also a complex structure on $\Ecal_{S,\textrm{sc}}$ if we set $j\Delta_Sf\doteq\Delta_S Jf$, where $f\in\Ycal$. We define the holomorphic and anti-holomorphic subspaces of $\Ycal^{\CC}$ as
\begin{align*}
\Zcal&\doteq\{(f-iJf)|f\in\Ycal\}\,,\\
\overline{\Zcal}&\doteq \{(f+iJf)|f\in\Ycal\}\,,
\end{align*}
respectively. Projections onto these subspaces are given by $\1_\Zcal=\frac{1}{2}(\1-iJ^{\CC})$ and $\1_{\overline{\Zcal}}=\frac{1}{2}(\1+iJ^{\CC})$

The CCR algebra corresponding to $(\Ycal,\Delta_S)$ is just the algebra $\tilde{\fA}_{S_0}$ introduced at the end of the previous section. Note that $\omega_{H,0}$ is a state on $\tilde{\fA}_{S}$ with covariance $H$. This state is pure if and only if the triple $(H,\Delta_S,J)$ is a  K\"ahler structure on $\mathcal{Y}$, i.e. all three structures are compatible and $\Delta_S\circ J=2H$. We can now decompose $\Delta_S^+$ in  the holomorphic basis.  A straightforward computation shows that
\[
\left<\1_{\overline{\Zcal}}f_1,\Delta_S^+(\1_{\Zcal}f_2)\right>=\left<f_1,\Delta_S^+f_2\right>\,,
\]
where $f_1,f_2\in\Ycal^{\CC}$ and remaining components vanish, so in the holomorphic basis $\Delta_S^+$ is represented by
\[
\begin{pmatrix}0&0\\
\Delta_S^+&0\end{pmatrix}\,,
\]
so it acts only on the holomorphic part of the first argument and the anti-holomorphic part of the second argument. 
\section{Time ordered products, and the perturbative construction of local nets}

In the previous section we were concerned only with the quantization of free theories (quadratic actions). Given an arbitrary action $S$ we first split $S=S_0+S_I$, where $S_0$ is quadratic. We already know how to quantize the classical model defined by $S_0$, so now is the time to introduce the interaction. We will do it in this section, following the ideas of  \cite{BP57,BS59,Hep66,EG,SR50,Ste71}, but before we start, we give a heuristic argument justifying our construction. The idea is to use the analogy with the interaction picture of quantum mechanics. Let $H_0$ be the Hamiltonian operator of the free theory and let  $H_{t,I}=-\int_K  \normOrd{\Lcal_I(0,\mathbf x)}d\sigma_t$ be the interaction Hamiltonian, where $\normOrd{\Lcal_I}$ is the normal-ordered Lagrangian density, constructed from the classical quantity ${\Lcal_I}$ and $K$ is some compact subset of $\Sigma$ (as explained in section \ref{canandP}). The rigorous ``smoothed-out'' version of the Hamiltonian quantization will be given in section \ref{Section:Hamiltonian}. 

We would like to use the Dyson formula and define the  time evolution operator as a time ordered exponential, i.e.

{\[U(t,s)=e^{itH_0}e^{-i(t-s)(H_0+H_I)}e^{-isH_0}=\]
\[1+\sum_{n=1}^\infty\frac{i^n}{n!}\int_ {([s,t]\times\mathbb R^3)^n}T(\normOrd{\Lcal_I(x_1)}\dots \normOrd{\Lcal_I(x_n)})d^{4n}x\,,\]}
where
\[x\mapsto \Lcal_I(x)=e^{iH_0x^0}\normOrd{\Lcal_I(0,\mathbf x)}e^{-iH_0x^0}
\]
is an operator-valued function and $T$ denotes time-ordering. Heuristically, one could use the unitary map defined above to obtain interacting fields as
\be\label{phiI}
\ph_I(x)=U(x^0,s)^{-1}\ph(x)U(x^0,s)=U(t,s)^{-1}U(t,x^0)\ph(x)U(x^0,s)\,,\ee
where $s<x^0<t$.

There are, however, serious problems with this heuristic formula. Firstly, typical Lagrangian densities, e.g.    
 $\normOrd{\Lcal_I(x)}=\normOrd{\ph(x)^4}$, 
can not be restricted to $\Sigma_0$ as operator valued distributions. This is the source of the so called UV problem. Moreover, as mentioned before, having the sharp cutoff function in the Lagrangian and Hamiltonian (like in \eqref{action}) leads to additional divergences  (St\"uckelberg divergences).
Finally there is the adiabatic limit problem related to the fact that the integral over $\mathbf x$ does not exist. Last but not least, the overall sum might not converge.
\subsection{Causal perturbation theory}
Some of these problems mentioned in the introduction can be easily dealt with by a slight modification of the above ansatz. For example, we avoid the St\"uckelberg divergences by replacing the sharp cutoffs with smooth test functions. The UV problem is solved by using causal perturbation theory in the sense of Epstein and Glaser \cite{EG}. In this method one switches the interaction on only in a compact region of spacetime and then takes the adiabatic limit (understood as a certain inductive limit) on the level of interacting observable algebras. These modifications of the Dyson formula ansatz lead to the definition of the \textit{formal S-matrix}:
\[S(g)=1+\sum_{n=1}^\infty\frac{i^n}{n!}\int g(x_1)\dots  g(x_n)T(\normOrd{\Lcal_I(x_1)}\dots \normOrd{\Lcal_I(x_n)})\,,\]
where $g$ is a test density. In order to make this formula well defined, we need to make sense of the time-ordered products of $\normOrd{\Lcal_I(x_i)}$. This will be done by Epstein-Glaser renormalization. Finally, the formula \eqref{phiI} has to be reinterpreted as a definition of a distribution, rather than a function. Hence, for a test density $f$ we obtain
\begin{align*}
\int f(x)\ph_I(x)&=S(g)^{-1}\sum_{n=0}^\infty \frac{i^n}{n!}\int f(x)g(x_1)\dots g(x_n)T\ph(x)\Lcal_I(x_1)\dots \Lcal_I(x_n)=\\
&=\frac{d}{d\lambda}\left.S(g)^{-1}S(g,\lambda f)\right|_{\lambda=0}\,,
\end{align*}
where $S(g,f)$ is the formal S-matrix with the Lagrangian density $g\Lcal_I+f\ph$. This is the so called \textit{Bogoliubov's formula} \cite{BS59}. 

We are now left with the problem of defining the time-ordered products on Wick-ordered quantities $\normOrd{\Lcal_I(x)}$. We have already mentioned in section \ref{Def:quant} that the normal ordering corresponds to passing between the star product $\star_H$ on $\fA_H$ and $\star$ on $\fA$. Note that elements of $\fA_H$ are functionals on $\Ecal_S$, so we can identify classical quantities in $\Fcal_{\loc}$ with quantum ones by means of $\TT^H_1:\Fcal_{\loc}\rightarrow \fA_H$ defined by $\TT^H_1=\id$. Composing with $\al_H^{-1}$ we obtain a map $\TT_1:\F\rightarrow \fA$, $\TT_1\doteq \al_H^{-1}\circ \TT^H_1$ which maps ``classical'' to ``quantum''. This map is interpreted as the normal ordering and we can now make an identification
\[
\normOrd{F}\doteq \TT_1F\,,\qquad F\in\Fcal_{\loc}\,.
\]
In the context of local covariance, this choice of normal ordering is not the most optimal one. This is because a family of Hadamard states cannot be chosen in a covariant way (i.e. compatible with embedding of globally hyperbolic spectimes), but a family of Hadamard parametrices can. The latter are bi-solutions of the linearized equations of motion only up to smooth terms. It is, therefore, more appropriate to define the normal ordering by a prescription where only the singular part of $H$ is subtracted from the correlation function of two fields, as opposed to the prescription where one subtracts the full $H$. Concretely, we set $\TT^H_1=\al_{w}$ so $\TT_1=\al^{\minus}_{H-w}$, where $w$ is the smooth part of the Hadamard 2-point function (see \cite{KW} and \cite{FR15} for a recent review). More precisely, this has to be understood as $\lim_{N\rightarrow \infty} \al^{\minus}_{H-w_N}F$ for $F\in\Fcal_{\loc}$ and this limit makes sense, because the series converges after finitely many steps. The function $w_N$ appearing in this prescription is $2N + 1$ times continuously
differentiable and it appears in the 2-point function as $\Delta_S^+=W_N^{sing}+w_N$. The singular part $W_N^{sing}$ is of the form ``$\frac{u}{\sigma}+v\ln\sigma$'', with $\sigma(x,y)$ denoting the square of the length of the geodesic connecting $x$ and $y$ and with geometrical determined smooth functions $u$ and $v$. For a more precise definition of what is the Hadamard form for of a 2-point function, see for example \cite{KW} or a recent review \cite{FR15}. 

More concretely, for a density of the form
\[
\Phi^{A,n}(f)(\ph)\doteq\int f(x)\frac{d^n}{d\la^n}A(x)(\ph)\big|_{\la=0}d^4x\,,
\]
where $A(x)(\ph)=e^{\la p(\nabla)\ph(x)}$ (here $p$ is a polynomial in covariant derivatives) we define
\[
\normOrd{\Phi^{A,n}}(f)\equiv {\TT_1}(\Phi^{A,n}(f)) \doteq \al_H^{\minus} \int f\,\frac{d^n}{d\la^n}A_H \big|_{\la=0}d^4x\,,
\]
where
\[
A_H(x)=e^{\frac12p(\nabla)\otimes p(\nabla)w_N(x,x)}A(x)\,,
\]

Unfortunately, the modifications which we have done so far do not render the time-ordered products well defined. Heuristically, we would like the time-ordered product of two functionals to be \eqref{tordered}, i.e.
\[
(F\T G)=e^{\hbar\left\langle\frac{\delta}{\delta \ph},\Delta_{S_0}^D \frac{\delta}{\delta\ph'}\right\rangle}F(\ph)G(\ph')|_{\ph'=\ph}\,,
\]
where $\Delta_{S_0}^D\doteq\frac12(\Delta_{S_0}^R+\Delta_{S_0}^A)$ is the Dirac propagator. This makes sense if both $F$ and $G$ are regular functionals (i.e. elements of $\Fcal_{\reg}$). This indeed provides the correct notion of time-ordering, since
\be\label{ordering}
F\T G=\left\{\begin{array}{rcl}
F\star G&\textrm{if}&\supp G\prec\supp F\,,\\
G\star F&\textrm{if}&\supp F\prec\supp G\,,
\end{array}\right.
\ee
where the relation ``$\prec$'' means ``not later than'' i.e. there exists a Cauchy surface which separates $\supp G$ and $\supp F$ and in the first case  $\supp F$ is in the future of this surface and in the second case it's in the past.

The time ordered product defined by \eqref{tordered} is associative, commutative and isomorphic to the point-wise product  by means of
\be\label{Tprod}
F\T G=\TT\left(\TT^{-1}F\cdot \TT^{-1}G\right)\,,
\ee
where 
\be
\TT=e^{i\hbar\langle\Delta_{S_0}^D,\frac{\delta^2}{\delta\ph^2}\rangle} 
\ee
or more precisely
\[
(\TT F)(\ph)\doteq \sum_{n=0}^\infty \frac{\hbar^n}{n!}\left<(i\Delta_{S_0}^D)^{\otimes n},F^{(2n)}(\ph)\right>\,.
\]
The linear operator $\Tcal$ defined above is sometimes called the ``time-ordering operator'' and it is interpreted as a map which goes from the ``classical'' to the ``quantum'', i.e.
\[
{(\Fcal_\reg,\cdot)\atop \textrm{classical}}\xrightarrow{\TT} {(\fA_{\reg},\star,\T)\atop \textrm{quantum}}\,,
\]
where $\fA_{\reg}\subset \fA$ is the range of $\TT$. Note that on the quantum side we have \textit{two} products. Using the time-ordered product we can express the formal $S$-matrix $\Scal:\F_\reg[[\hbar]]\rightarrow\F_\reg[[\hbar]]$ as the time ordered exponential:
\be\label{Smatrix}
\Scal(V)\doteq e_{\sst{\TT}}^{iV/\hbar}=\TT\big( e^{\TT^{-1}iV/\hbar}\big) \,.
\ee
According to our interpretation of $\TT$, $\Scal$ is a map on the ``quantum'' algebra $\fA$ to itself. Interacting fields are obtained by means of the Bogoliubov formula, which reads
\begin{align}\label{RV}
R_V(F)&= -i\hbar\frac{d}{d\la}\left.\left(\Scal(V)^{\star\minus}\star\Scal(V+\la F)\right)\right|_{\lambda=0}=\nonumber\\
&=\left(e_{\sst{\TT}}^{i V/\hbar}\right)^{\star\minus}\star\left(e_{\sst{\TT}}^{i V/\hbar}\T  F\right)\,.
\end{align}
We interpret $R_V(F)$ as the interacting quantity corresponding to $F$. We can also define the interacting star product as
\[
F\star_V G\doteq  R_V^{-1}(R_VF\star R_V G)\,.
\]
The interacting theory is given in terms of the algebra $(\Fcal_{\reg},\star_V )$ and $R_V$ acts as the intertwining map between the free quantum theory and the interacting quantum theory, i.e.
\be\label{diagreg}
{(\Fcal_\reg,\cdot)\atop \textrm{classical}}\xrightarrow{\TT} {(\fA_{\reg},\star,\T)\atop {\textrm{free}\atop \textrm{quantum}}}\xrightarrow{R_V^{-1}} {(\fA_{\reg},\star_V)\atop {\textrm{interacting}\atop \textrm{quantum}}}\,.
\ee

All these formulas make sense if we restrict ourselves to regular functionals. This is, however, not satisfactory for our purposes, since typical interactions are local and non-linear, hence not regular. In the first attempt we could try to pass to a different star product, which amounts to replacing $\Delta_{S_0}$ by $\Delta_{S_0}^+$ and $\Delta_{S_0}^D$ by the Feynman propagator $\Delta_{S_0}^F=i\Delta_{S_0}^D+H$, so our diagram gets modified to
\[
{(\Fcal_\reg,\cdot)\atop \textrm{classical}}\xrightarrow{\TT^H}(\fA_{\reg},\star_H,\cdot_{\TT^H})\xrightarrow{\al_H^{-1}}{(\fA_{\reg},\star,\T)\atop {\textrm{free}\atop \textrm{quantum}}}\xrightarrow{R_V^{-1}} {(\fA_{\reg},\star_V)\atop {\textrm{interacting}\atop \textrm{quantum}}}\,,
\]
where $\TT^H\doteq e^{i\hbar\langle\Delta_{S_0}^F,\frac{\delta^2}{\delta\psi^2}\rangle}$, so $\TT=\al_H^{-1}\circ\TT^H$. This modification of the formalism, however, doesn't solve the problem yet. To extend our formalism to arbitrary local functionals, we need to perform the renormalization. The difficulty which we have to face is the fact that the WF set of  $\Delta_{S_0}^F$ at 0 is like the WF set of the Dirac delta and therefore the tensor powers of $\Delta_{S_0}^F$ cannot be contracted with derivatives of local functionals.

However, there is a way to extend $\TT^{\sst H}$ to \textit{local} functionals. First we extent $\TT^{\sst H}$ to $\F_{\loc}$ by setting $\TT^{\sst H}=\TT_1^{\sst H}$. We discuss here only the Minkowski spacetime situation, so we can set $\TT_1^H=\id$. The subspace $\TT^{H}(\Fcal_{\loc})\subset\fA$ will be denoted by $\fA^H_{\loc}$. Let us define the n-th order time-ordered product as
\[
\TT^{\sst H}_n(F_1,\dots,F_n)\doteq F_1{\cdot_{\TT_H}}\dots {\cdot_{\TT_H}} F_n,
\]  
whenever it exists. It is well defined for $F_1,\dots,F_n\in \Fcal_{\loc}$ with pairwise disjoint supports and we will denote this domain of definition by $(\Fcal_{\loc})^{\otimes n}_{\mathrm{pds}}$. Moreover 
\be\label{causalfact}\tag{\bf T 1}
\TT^{\sst H}_n(F_1,\dots,F_n)=\TT^{\sst H}_k(F_1,\dots,F_k)\star_H \TT^{\sst H}_{n-k}(F_{k+1},\dots,F_n)\,,
\ee
 if the supports $\supp F_i$, $i=1,\dots,k$ of the first $k$ entries do not intersect the past of the supports $\supp F_j$, $j=k+1,\dots,n$ of the last $n-k$ entries. This property is called the \textit{causal factorisation property}. We will take it as an axiom that we want to impose while extending time-ordered products to arbitrary local arguments.  The other axioms include
 \begin{enumerate}[(\bf T 1)]
 \setcounter{enumi}{1}
\item {\bf Starting element}: $\TT^{\sst H}_0=1$, $\TT^{\sst H}_1=\id$,
\item {\bf Symmetry}: Each $\TT^{\sst H}_n$ is symmetric (graded symmetric if Fermions are present).
\item {\bf $\ph$-Locality}: $\TT^{\sst H}_{n}(F_1,\ldots,F_n)$, as a functional on $\Ecal$, depends on $\ph$ only via the functional derivatives of $F_1,\ldots,F_n$.
\label{Tf}
\end{enumerate} 
In the seminal paper \cite{EG}, Epstein and Glaser have shown that such a family of maps exists and non-uniqueness in  defining $\TT_n^{\sst H}$'s is fully absorbed  into adding multilinear maps $Z_n: \fA^{\otimes n}_{\loc}\rightarrow\fA_{\loc} $, i.e.
 \[
 \widetilde{\TT^{\sst H}}_n(F_1,\ldots,F_n)= \TT^{\sst H}_n(F_1,\ldots,F_n)+\Zcal_n(F_1,\ldots,F_n)\,,
 \]
 where $\{\TT^{\sst H}_n\}_{n\in\NN}$ and  $\{\widetilde{\TT^{\sst H}}_n\}_{n\in\NN}$ are two choice of time-ordered products that coincide
up to order $n-1$. The renormalized S-matrix is now defined by
\[
\Scal(V)=\sum_{n=0}^\infty \tfrac{1}{n!}\TT_n(V,\ldots,V)=\sum_{n=0}^\infty \tfrac{1}{n!}\al_{\sst H}^{\minus}\circ\TT^{\sst H}_n(\al_{\sst H}V,\ldots,\al{\sst H}V)\,.
\]
The causal factorisation property for time ordered products implies that the S-matrix satisfies Bogoliubov's factorization relation
\begin{equation}\label{factorisation}
{\Scal(V_1+V_2+V_3)=\Scal(V_1+V_2)\Scal(V_2)^{-1}\Scal(V_2+V_3)} 
\end{equation}
if the support of $V_1$ does not intersect the past of the support of $V_3$.

We can also define the renormalized map $\TT:\Fcal\rightarrow\fA$ by $\TT\doteq\bigoplus_n\al_{\sst H}^{\minus}\circ\TT^{\sst H}_n\circ m^{-1}$, where $m^{-1}:\Fcal\to S^\bullet\Fcal^{(0)}_\loc$ is the inverse of the multiplication, as defined in \cite{FR13} and $\Fcal^{(0)}_\loc$ is the space of local functionals that vanish at 0. The renormalized time ordered product $\T$ is now a binary operation defined on the domain $D_{\TT}\doteq \TT(\Fcal)$. Analogously to the diagram \eqref{diagreg}, we obtain now
\be\label{diagloc}
{(\Fcal,\cdot)\atop \textrm{classical}}\xrightarrow{\TT} {(\fA,\star,\T)\atop {\textrm{free}\atop \textrm{quantum}}}\xrightarrow{R_V^{-1}} {(\fA,\star_V)\atop {\textrm{interacting}\atop \textrm{quantum}}}\,,
\ee
with the caveat that $\T$ is well defined on $D_{\TT}\subset \fA$.

We will now discuss in detail the ambiguity arising in defining $\TT_n$'s. In physics this is known as the \textit{renormalization ambiguity}. To understand it better and to relate it with the notion of the \textit{renormalization group}, we first define a map $\Zcal:\fA_\loc[[\hbar]]\rightarrow\fA_\loc[[\hbar]]$ by summing up all the $\Zcal_n$'s relating two chosen prescriptions to define the time-ordered products. For any two choices of $\TT_n$'s the corresponding map $\Zcal$ has the following properties:
\begin{enumerate}[{(\bf Z 1)}]
\item $\Zcal(0)=0$,\label{Zs}
\item $\Zcal^{(1)}(0)=\id$,
\item $\Zcal=\id+\Ocal(\hbar)$,
\item $\Zcal(F+G+H)=\Zcal(F+G)+\Zcal(G+H)-\Zcal(G)$, if $\supp\,F\cap\supp\,G$,
\item $\frac{\delta \Zcal}{\delta\ph}=0$.\label{Zf}
\end{enumerate}
The group of formal diffeomorphisms of $\fA_\loc[[\hbar]]$ that fulfill {(\bf Z \ref{Zs})} -- {(\bf Z \ref{Zf})} is called the  \textit{St{\"u}ckelberg-Petermann renormalization group} $\Rcal$. There is a relation between the formal S-matrices and elements of $\Rcal$ provided by the main theorem of renormalization \cite{DF04,BDF09}. It states that for two S-matrices  $\Scal$ and $\hat{\Scal}$, built from time ordered products  satisfying the axioms ({\bf T 1})
-- ({\bf T~\ref{Tf}}), there exists $\Zcal\in\Rcal$ such that
\be
\hat{\Scal}=\Scal\circ \Zcal\,, 
\ee
where $\Zcal\in\Rcal$ and conversely, if $\Scal$ is an S-matrix satisfying the axioms ({\bf T 1})
-- ({\bf T~\ref{Tf}}) and $\Zcal\in\Rcal$ then also $\hat{\Scal}$ fulfills the axioms.
\subsection{Methods for explicit construction of time-ordered products}\label{explicit: constr}
The proof of existence of time-ordered products with properties  ({\bf T 1})
-- ({\bf T~\ref{Tf}}) given in \cite{EG} is rather abstract and relies on an inductive argument. For practical purposes an existence result is not sufficient and one would like to obtain some explicit formulas for $\TT_n$'s. In this section we will review results which show that the problem of constructing time-ordered products reduces to extending certain distributions. Subsequently, we will give some concrete computational prescriptions for constructing such extensions.

We start with an example. 
Let $F=\frac12\int \ph^2fd\mu$, $G=\frac12\int \ph^2gd\mu$, $f,g\in\Dcal$. If $\supp g\cap \supp f=\varnothing$, then the time ordered product $\cdot_{\TT}$ of $F$ and $G$ is given by
\begin{multline*}\label{ex:DeltaF}
\TT_2(F,G)(\ph)=\\=(F\cdot_{\TT}G)(\ph)=F(\ph)G(\ph)+i\hbar \int \ph(x)\ph(y)f(x)g(y)\Delta_{S_0}^F(x,y)d\mu(x)d\mu(y)-\frac{\hbar^2}{2}\int \Delta_{S_0}^F(x,y)^2 f(x)g(y)d\mu(x)d\mu(y)\ .
\end{multline*}
In the least term of the expression above we have a pointwise product of a distribution with itself. This could potentially cause problems. If $x\not=y$, then if $(x,k)$ and $(y,-k')$ belong to the wave front set, then $k,-k'$ are cotangent to a null geodesics connecting 
$x$ and $y$. Moreover, $k$ is future directed if $x$ is in the future of $y$ and past directed otherwise, so the sum of two such covectors doesn't vanish. Hence, the condition on the multiplicability of distributions presented in section \ref{sec:Functional derivatives} implies that $(\Delta_{S_0}^F)^2$ as a distribution is well defined on the complement of the diagonal $\{(x,x)|x\in M\}$. Let us now consider what happens on the diagonal. There, the only restriction is $k=-k'$, hence the sum of $\WF(\Delta_{S_0}^F)$ with itself contains the zero section of the cotangent bundle at the diagonal. The problem of defining $\TT^{\sst H}_2(F,G)$ reduces now to the problem of extension of $\Delta_{S_0}^F$ to a distribution defined everywhere.

This generalizes, and the construction of $\TT^{\sst H}_n$'s reduces to extending numerical distributions defined everywhere outside certain subdiagonals in $M^n$. The construction proceeds recursively and, having constructed the time-ordered products of order $k<n$, at order $n$ one is left with the problem of extending a distribution defined everywhere outside the thin diagonal of $M^n$. On Minkowski spacetime, exploiting the translational symmetry of $\MM$, this reduces to extending a numerical distribution defined everywhere outside $0$. One way of constructing explicitly such distributional extensions relies on the so called splitting method (see for example \cite{Scharf}). Here we will take a different approach, based on the notion of Steinmann's scaling degree \cite{Ste71}. Here is the definition:

\begin{definition}
Let $U\subset \RR^n$ be a scale invariant open subset (i.e. $\lambda U=U$ for $\lambda>0$), and let $t\in\Dcal'(U)$ be a distribution on $U$. Let
$t_{\la}(x)=t(\la x)$ be the scaled distribution. The scaling degree $\mathrm{sd}$ of $t$ is 
\be
\mathrm{sd}\,t=\mathrm{inf}\{\delta\in\RR|\lim_{\la\to0}\la^{\delta}t_{\la}=0\} \ .
\ee  
\end{definition} 
The degree of divergence, another important concept used often in the literature, is defined as:
\[
\mathrm{div}(t)\doteq \mathrm{sd}(t)-n\,.
\]
The crucial result which allows us to construct time-ordered products is stated in the following theorem:
\begin{theorem}\label{extension}
Let $t\in \Dcal(\RR^n\setminus\{0\})$ with scaling degree $\mathrm{sd}\,t<\infty$. Then there exists an extension of $t$ to an everywhere defined distribution with the same scaling degree. The extension is unique up to the addition of a derivative $P(\partial)\delta$ of the delta function, where $P$ is a polynomial with degree bounded by 
$\mathrm{div}(t)$ (hence vanishes for $\mathrm{sd}\,t<n$).  
\end{theorem}
In the example presented at the beginning of this subsection, the scaling degree of $(\Delta_{S_0}^F)^2$ in 4 dimensions is 4, so the extension exists and is unique up to the addition of a multiple of the delta function.

The result above allows in principle to extend all the numerical distributions we need for the construction of time-ordered products. However, the computations can in general get very complicated, so it is convenient to formulate the combinatorics underlying our construction in terms of Feynman graphs. In the pAQFT framework, these are not fundamental objects, but instead they are derived (together with the corresponding Feynman rules) from time-ordered products.

 Time-ordered products $\TT^{\sst H}_n$ should be maps from $\Fcal_\loc^{\otimes n}$ to $\Fcal_{\mc}[[\hbar]]$ and, as indicated in the previous section, they are obtained by extending non-renormalized expressions that are originally defined only on $(\Fcal_\loc)^{\otimes n}_{\mathrm{pds}}$. Let us consider $F\equiv F_1\otimes\dots\otimes F_n\in(\Fcal_\loc)^{\otimes n}_{\mathrm{pds}}$ with the corresponding Wick-ordered quantities are elements of $\fA_{\loc}$ given by $A_1\doteq\TT F_1,\ldots,A_n\doteq\TT F_n\in\Fcal_{\loc}$.  Note that $F$ induces a map from $\Ecal^{n}$ to $\RR$ by $F(\ph_1,...,\ph_2)=F_1(\ph_1)\cdots F_n(\ph_n)$. When we talk about functionals on $\Ecal$ we will denote the variable by $\ph$ and for functionals on $\Ecal^{n}$ we take an $n$-tuple $(\ph_1,...,\ph_n)$.
 
Let us denote $D_{ij}\doteq i\hbar\langle\Delta_{S_0}^F,\frac{\delta^2}{\delta\ph_i\delta\ph_j}\rangle$ and $D\doteq i\hbar \langle\Delta_{S_0}^F,\frac{\delta^2}{\delta\ph^2}\rangle$. The Leibniz rule for differentiation can be formulated as
\be
\frac{\delta}{\delta\ph}\circ m_n=m_n\circ\left(\sum_{i=1}^n\frac{\delta}{\delta\ph_i}\right) \,,
\ee
where $m_n$ is the pointwise multiplication of $n$ arguments, or in other  words, a pullback through the diagonal map $\Ecal\rightarrow \Ecal^n$, $\ph\mapsto (\ph,\ldots,\ph)$. The Leibniz rule implies that the non-renormalized expression for $\TT^{\sst H}$ satisfies
\[
\TT^{\sst H}\circ m_n=e^{\frac{D}{2}}\circ m_n=m_n\circ e^{\sum_{i<j}D_{ij}+\sum_i \frac{1}{2}D_{ii}}\,,
\]
Hence
\[
F_1\cdot_{\TT^{\sst H}}\dots\cdot_{\TT^{\sst H}} F_n=e^{\frac{D}{2}}\circ m_n(e^{- \frac{1}{2}D_{11}}F_1,\dots,e^{- \frac{1}{2}D_{nn}}F_n)=m_n\circ e^{\sum_{i<j}D_{ij}}(F_1,\dots F_n)\equiv m_n\circ T_n(F_1,\dots F_n)\,.
\]
We can now use an identity
\be\label{expDij}
e^{\sum_{i<j}D_{ij}}=\prod_{i<j}\sum_{l_{ij}=0}^{\infty}\frac{D_{ij}^{l_{ij}}}{l_{ij}!}
\ee
to express time ordered products in terms of graphs.  Let $\Gcal_n$ be the set of all graphs with vertex set $V(\Gamma)=\{1,\dots n\}$ and $l_{ij}$ the number of lines $e\in E(\Gamma)$ connecting the vertices $i$ and $j$. We set $l_{ij}=l_{ji}$ for $i>j$ and $l_{ii}=0$. If $e$ connects $i$ and $j$ we set $\partial e:=\{i,j\}$. Then 
\be\label{time:ord}
T_n=\sum_{\Gamma\in \Gcal_n}T_{\Gamma}\,,
\ee
where
\be\label{GraphDO}
T_{\Gamma}=\frac{1}{\textrm{Sym}(\Gamma)}\langle t_{\Gamma},\delta_{\Gamma}\rangle\,,
\ee
with
\[\delta_{\Gamma}=\frac{\delta^{2\,|E(\Gamma)|}}{\prod_{i\in V(\Gamma)}\prod_{e:i\in\partial e}\delta\ph_i(x_{e,i})}\]
and
\be\label{SGamma}
t_{\Gamma}=\prod_{e\in E(\Gamma)}\hbar\Delta_F(x_{e,i},i\in\partial e)
\ee
The, so called, symmetry factor $\textrm{Sym}$ is the number of possible permutations of lines joining 
the same two vertices, $\textrm{Sym}(\Gamma)=\prod_{i<j}l_{ij}!$. Note that $T_{\Gamma}$ is a map from $(\Fcal_\loc)_{\mathrm{pds}}^{\otimes V}$ to $\Ci(\Ecal^{|V|},\RR)[[\hbar]]$, where $\otimes V$ means that the factors in the tensor product are numbered by vertices and to a vertex $v\in V(\Gamma)$ we assign the variable $\ph_v$. The renormalization problem is now the problem to extend $T_n$'s to maps on $(\Fcal_\loc)^{\otimes n}$ and this can be achieved by extending all the maps  $T_{\Gamma}$ and using formula \eqref{time:ord}.

First we note that  functional derivatives of local functionals are of the form
\be\label{dF0}
F^{(l)}(\ph)(x_1,\dots,x_l)=\int\sum_{j=1}^Ng_j[\ph](y)p_j(\partial_{x_1},\dots,\partial_{x_l})\prod_{i=1}^l\delta(y-x_i)d\mu(y)\,,
\ee
where $N\in\NN$, $p_j$'s are polynomials in partial derivatives and $g_j[\ph]$ are $\ph$-dependent test functions. The representation above is not unique, since some of the partial derivatives $\partial_{x_i}$ can be replaced with $\partial_y$ and applied to $g_j[\ph]$. Another representation of  $F^{(l)}(\ph)$ is obtained by performing the integral above and using the centre of mass and relative coordinates:
\be\label{dF}
F^{(l)}(\ph)(x_1,\dots,x_l)=\sum_\bet f_\bet[\ph](z)\partial^{\bet}\de(x^{\text{rel}})
\ee
where $\beta\in\NN_0^{4(l-1)}$, test functions 
$f_\bet[\ph](x)\in\Dcal$ are now $\ph$-dependent functions of the center of mass coordinate $z=(x_1+\dots+x_k)/k$ and $x^{\text{rel}}=(x_1-z,\dots, x_k-z)$ denotes the relative coordinates.

Using \eqref{dF0} we see that the functional differential operator $\delta_\Gamma$ applied to $F\in\mathcal{F}_{\mathrm{loc}}^{\otimes n}$ yields, at any $n$-tuple of field configurations $(\varphi_1,\dots,\varphi_n)$, a compactly supported distribution in the variables $x_{e,i},i\in\partial e, e\in E(\Gamma)$ with support on the partial diagonal $\Delta_{\Gamma}=\{x_{e,i}=x_{f,i},i\in\partial e\cap\partial f, e,f\in E(\Gamma)\}\subset \MM^{2|E(\Ga)|}$ and with a wavefront set perpendicular to $T\Delta_{\Gamma}$.  Note that the partial diagonal $\Delta_{\Gamma}$ can be parametrized using the center of mass coordinates
\[
z_v\doteq \frac{1}{\textrm{valence}(v)}\sum_{e|v\in\partial e} x_{e,v}\,,
\]
assigned to each vertex. The remaining relative coordinates are $x_{e,v}^{\text{rel}}=x_{e,v}-z_v$, where $v\in V(\Ga)$, $e\in E(\Ga)$ and $v\in\partial e$. Obviously, we have $\sum_{e|v\in\partial e} x_{e,v}^{\text{rel}}=0$ for all $v\in V(\Ga)$. In this parametrization $\delta_\Gamma F$ can be written as a finite sum
\[
\delta_\Gamma F=\sum_{\textrm{finite}}f^\beta\partial_\beta\delta_{\textrm{rel}}\,,
\]
where $\beta\in\NN_0^{4|V(\Gamma)|}$, each $f^{\beta}(\ph_1,...,\ph_n)$ is a test function on $\Delta_{\Gamma}$ and $\delta_{\textrm{rel}}$ is the Dirac delta distribution in relative coordinates, i.e. $\delta_{\textrm{rel}}(g)=g(0,\ldots,0)$, where $g$ is a function of $(x_{e,v}^{\textrm{rel}},v\in V(\Ga), e\in E(\Ga))$.

 We can simplify our notation even further. Let $Y_\Gamma$ denote the vector space spanned by derivatives of the Dirac delta distributions $\partial_\beta\delta_{\textrm{rel}}$, where $\beta\in\NN_0^{4|V(\Gamma)|}$. Obviously, $Y_\Gamma$ is graded by $|\beta|$. Let $\mathcal{D}(\Delta_\Gamma, Y_\Gamma)$ denote the graded space of test functions on $\Delta_\Gamma$ with values in $Y_\Gamma$. With this notation we have $\delta_\Gamma F\in\mathcal{D}(\Delta_\Gamma, Y_\Gamma)$ and if $F\in(\mathcal{F}_{\mathrm{loc}})_{\mathrm{pds}}^{\otimes n}$, then $\delta_\Gamma F$ is supported on  $\Delta_{\Gamma}\setminus\DIAG$, where $\DIAG$ is the large diagonal:
 \[
\DIAG=\left\{ z\in\Delta_{\Gamma}|\,\exists v,w\in V(\Gamma),v\neq w:\, z_{v}=z_{w}\right\} \,.
\]
We can now write \eqref{GraphDO}  in the form
\[
\frac{1}{\textrm{Sym}(\Gamma)}\langle t_{\Gamma},\delta_{\Gamma}\rangle=\sum_{\textrm{finite}}\left<f^\beta\partial_\beta\delta_{\textrm{rel}},t_\Gamma\right>
\]
where $t_{\Gamma}$ is now written in terms of centre of mass and relative coordinates. To see that this expression is well defined, note that we can move all the partial derivatives $\partial_\beta$ to $t_{\Ga}$ by formal partial integration. Then the contraction with $\delta_{\textrm{rel}}$ is just the pullback through the diagonal map 
map $\rho_{\Ga}:\Delta_\Gamma\rightarrow\MM^{2|E(\Ga)|}$ by
\[
(\rho_{\Ga}(z))_{e,v}=z_v\,\quad\mathrm{if}\,v\in\partial e\,.
\]
From the wavefront set properties of $\Delta^F_{S_0}$, we deduce that the pullback $\rho_{\Ga}^*$ of each $ t_{\Ga}^\beta\doteq\partial_\beta t_{\Ga}$ is a well defined distribution on
$\Delta_\Ga\backslash\DIAG$, so  \eqref{GraphDO} makes sense if $F\in(\mathcal{F}_{\mathrm{loc}})_{\mathrm{pds}}^{\otimes n}$, as expected. 
We conclude that $t_{\Ga}\in \mathcal{D}'(\Delta_\Gamma\backslash\DIAG, Y_\Gamma)$, where the duality between $t_{\Ga}$ and a test function $f=\sum_{\textrm{finite}}f^\beta\partial_\beta\delta$ is given by
\[
\langle t_{\Gamma},f\rangle\doteq\sum_\beta\langle t^\beta_\Gamma,f_\beta\rangle\,.
\]
The renormalization problem now reduces to finding the extensions of $ t^\beta_\Gamma$, so that $t^\beta_\Gamma$ gets extended to an element of $\mathcal{D}'(\Delta_\Gamma, Y_\Gamma)$. The solution to this problem is obtained by using the inductive procedure of Epstein and Glaser. The induction step works as follows: if 
$t_{\Gamma'}$ is known for all graphs $\Gamma'$ with fewer vertices than $\Gamma$, then $t_\Gamma$ can be uniquely defined for all \textit{disconnected}, all \textit{connected one particle reducible} and all \textit{one particle irreducible one vertex reducible graphs}. Graphs which are irreducible and do not contain any non-trivial   irreducible subgraphs are called \textit{EG-primitive}. For the remaining graphs, called \textit{EG-irreducible}, $t_{\Gamma}$ is defined uniquely on all $f\in\mathcal{D}(\Delta_{\Gamma},Y_\Gamma)$ of the form above where $f_\beta$ vanishes together with all its derivatives of order  $\leq\omega_\Gamma+|\beta|$ on the thin diagonal of $\Delta_\Gamma$. Here
\[\omega_\Gamma=(d-2)|E(\Gamma)|-d(|V(\Gamma)|-1)\]
is the degree of divergence of the graph $\Gamma$. We  denote this subspace by $\mathcal{D}_{\omega_{\Gamma}}(\Delta_\Gamma,Y_\Gamma)$. Graphs which are irreducible and do not contain any non-trivial   irreducible subgraphs are called \textit{EG-primitive}.
 Renormalization amounts to project a generic $f$ to this subspace by a translation invariant projection $W_{\Gamma}:\mathcal{D}(\Delta_\Gamma,Y_\Gamma)\to\mathcal{D}_{\omega_\Gamma}(\Delta_\Gamma,Y_\Gamma)$. Different renormalization schemes differ by different choices of the projections $W_\Gamma$ (see \cite{DFKR14} for details). 

On Minkowski spacetime we have further simplifications. By exploiting the translation invariance we find that, at each step of the recursive construction of time-ordered products, the renormalization problem reduces to the problem of extension of some distribution defined everywhere outside the origin, so this is what we will focus on now.

For concrete computations it is convenient to construct these extensions with the use of \textit{regularization}.
Let us first define the notion of a \textit{regularization of a distribution}. Let $\tilde{t}\in\Dcal'(\RR^d\setminus\{0\})$, $d\in\NN$, be a distribution with degree of 
divergence $\omega$, and by $\bar{t}\in\Dcal_\omega'(\RR^d)$ we denote the unique extension of $\tilde{t}$ with the same degree of divergence. A family of distributions $\{t^\zeta\}_{\zeta\in\Omega\setminus\{0\}}$, $t^\zeta\in\Dcal'(\RR^d)$, with $\Omega\subset\CC$ a neighborhood of the origin, is called a regularization of $\tilde{t}$, if
\be\label{eq:regularization}
\forall g\in\Dcal_\lambda(\RR^d):\quad\lim_{\zeta\rightarrow0}\langle t^\zeta,g\rangle=\langle \bar{t},g\rangle\,.
\ee
We say that the regularization $\{t^\zeta\}$ is called analytic, if for all functions $f\in\Dcal(\RR^n)$ the map
\be
\Omega\setminus\{0\}\ni\zeta\mapsto \langle t^\zeta,f \rangle
\ee
is analytic with a pole of finite order at the origin. The regularization $\{t^\zeta\}$ is called finite, if 
the limit $\lim_{\zeta\rightarrow 0}\langle t^\zeta,f\rangle\in\CC$ exists $\forall f\in\Dcal(\RR^d)$.

For a finite regularization the limit
$\lim_{\zeta\rightarrow0}t^\zeta$ is, as expected, a solution $t$ of the
extension (renormalization) problem. Given a regularization $\{t^\zeta\}$ of $t$, it follows from \eqref{eq:regularization} 
that for any projection $W:\Dcal\rightarrow\Dcal_\omega$
\be\label{regW-1}
\langle \bar{t},Wf\rangle=\lim_{\zeta\rightarrow0}\langle t^\zeta,Wf\rangle\, \quad \forall f\in\Dcal(\RR^n)\,.
\ee
It was shown in \cite{DF04} that any extension $t\in\Dcal'(\RR^d)$ of $\tilde t$ with the same scaling degree is of the form $\langle t,f\rangle=\langle \bar t,Wf\rangle$ with some $W$-projection of the form 
\be\label{W:proj}
Wf:=f-\sum_{\left|\alpha\right|\leq\lambda}f^{\left(\alpha\right)}(0)\, w_{\alpha}\,,
\ee
where  $w_{\alpha}\in\Dcal(\RR^d)$ such that for all  multiindices $\beta\in\NN_0^d$ with $\left|\beta\right|\leq\omega$ we have $\partial^{\beta}w_{\alpha}(0)=\delta_{\alpha}^{\beta},\,\left|\alpha\right|\leq\omega$
Hence
\be\label{regW-2}
\langle \bar{t},Wf\rangle=\lim_{\zeta\rightarrow0}\left[\langle t^\zeta,f\rangle - \sum_{|\al|\leq\sd(t)-n}\langle t^\zeta,w_\al\rangle\; f^{(\al)}(0)\right].
\ee
In general, we cannot split the limit on the right hand side
into two well defined terms. However, if the regularization $\{t^\zeta,\zeta\in\Omega\setminus\{0\}\}$ is analytic, then we can expand each term into a Laurent series around $\zeta=0$, and because the overall limit is finite, the principal parts ($\pp$) of these two Laurent series must be the same.
This means that the principal part of any analytic regularization $\{t^\zeta\}$ of a distribution $t\in\Dcal'(\RR^d\setminus\{0\})$ is a local distribution of order $\sd(t)-d$. Following \cite{DFKR14}, we can now give a definition of the minimal subtraction in the EG framework.
\begin{definition}[Minimal Subtraction]\label{cor:MS-same-sd}
The regular part ($\rp=1-\pp$) of any analytic regularization $\{t^\zeta\}$ of a distribution $\tilde{t}\in\Dcal'(\RR^d\setminus\{0\})$ defines by
\be\label{def:MS}
\langle t^\MS,f\rangle :=\lim_{\zeta\rightarrow0} \rp(\langle t^\zeta,f\rangle)
\ee
an extension of $\tilde{t}$ with the same scaling degree, $\sd(t^\MS)=\sd(\tilde{t})$.
The extension $t^\MS$ defined by (\ref{def:MS}) is called the ``minimal subtraction''.
\end{definition}
\subsection{Interacting theories}
\label{sec:Time ordered products}
Let us now discuss the problem of constructing interacting nets of observables. We start from a space $\cD^n$ of functions $f:\MM\to\RR^n$ with compact support. We assume that we have unitaries {$S(f)$}, $f\in\cD^n$
 with $S(0)=0$, which generate a *-subalgebra {$\tilde{\fA}$} of {$\fA$} and satisfy for $f,g,h\in\cD$  Bogoliubov's factorization relation
\begin{equation}\nonumber
{S(f+g+h)=S(f+g)S(g)^{-1}S(g+h)} 
\end{equation}
if the past {$J_-$} of {$\supp h$} does not intersect  {$\supp f$}
(or, equivalently, if the future {$J_+$} of {$\supp f$} does not intersect  {$\supp h$}). We can obtain these as formal S-matrices $S(f)\doteq\Scal(V(f))$, discussed in the previous section (see property \eqref{factorisation}), for a generalized Lagrangian $V(f)=\al_{\sst H}\left(\sum_{j=1}^n \int A_j(x)f^j(x)d\mu(x)\right)$, where $f\in\Dcal^n$ and each $A_j(x)$ is a local function $\ph\in\Ecal$. Typically $A_j'$ are polynomial and they represent Lagrangian densities of various interaction terms that one can add to the free action $S_0$.

We also assume that the translation group of Minkowski space acts by automorphisms $\al_x$  on $\tilde{\fA}$ such that
\begin{equation}\nonumber
\al_x(S(f))=S(f_x)\ ,\ f_x(y)=f(y-x)\ .
\end{equation}
Obviously, this is also satisfied for the S-matrices discussed so far. 
Under these general assumptions, we define local algebras {$\fA(\cO)$}, $\Ocal\subset\MM$, as the *-subalgebras of $\fA$ generated by $S(f),\supp f\subset\cO$ 
and obtain a translation covariant Haag-Kastler net on Minkowski space. To justify this claim, we will now check that all the axioms are satisfied.

{\em Isotony} and {\em Covariance} are obvious, 
and {\em Locality} follows from the fact that 
for functions $f,g$ with spacelike separated supports
\begin{equation}
\supp f\cap J_\pm(\supp g)=\varnothing
\end{equation}
and hence 
\begin{equation}
S(f)S(g)=S(f+g)=S(g)S(f)\ .
\end{equation}

The crucial observation is now that  the map $f\mapsto S(f)$ induces a large family of objects that satisfy Bogoliubov's factorisation relation, which are labeled by test functions $g\in\cD^n$, namely the {\em relative S-matrices}
\begin{equation}\nonumber
{f\mapsto S_g(f)=S(g)^{-1}S(g+f)}\ .
\end{equation}
We can choose $A_0(x)=\cL_I(x)$ to be the Lagrangian density of the interaction term. Then, for $g=(g_0,0,\ldots,0)$, we obtain $V(g)=\int \cL_Ig_0d\mu\equiv L_I(g_0)$, where $g_0\in\Dcal$.
Note that $S(g+\la f)=\Scal(\al_{\sst H}( \cL_I(g_0)+\la\sum_j\int A_jf_jd\mu))$, so the derivative of $S$ with respect to $\la$ is just the retarded field $R_{ \cL_I(g_0)}(V(f))$. Let us now prove that the causal factorisation property indeed holds for $S_g(f)$.
 Let $f,h\in \cD^n$ such that $\supp f$ does not intersect $J_-(\supp h)$. Let $g,g'\in\cD^n$. Then
\begin{align*}
S_g(f+g'+h)&= S(g)^{-1}S(f+(g+g')+h)\\
&= S(g)^{-1}S(f+(g+g'))S(g+g')^{-1}S((g+g')+h)\\
&= S_g(f+g')S_g(g')^{-1}\underbrace{S(g)^{-1}S(g)}_{=1}S_g(g'+h)\ .\ \ \ \square
\end{align*}

We consider {$S_g(f)$} as the retarded observable $S(f)$ under the influence of the interaction  $L_I(g_0)$. The Haag-Kastler net {$\fA_g$} of the interacting theory is then defined by
the local algebras {$\fA_g(\cO)$} which are generated by the relative S-matrices $S_g(f),\supp f\subset \cO$. These can indeed be interpreted as retarded observables, as $S_g(f)$ depends only on the behavior of $g$ in the {past} of {$\supp f$}. More precisely, $\supp (g-g')\cap J_-(\supp f)=\varnothing$ implies
\begin{align*}S_g(f)&=S(g)^{-1}S((g-g')+g'+f)=\\
&=S(g)^{-1}S((g-g')+g')S(g')^{-1}S(g'+f)=S_{g'}(f)\ .
\end{align*}
The second observation is that $S_g(f)$ depends on the behavior of $g$ outside of the {future} of {$\supp f$} via a (formal) unitary transformation which does not depend on $f$. Namely,
$\supp (g-g')\cap J_+(\supp f)=\varnothing$ implies
\begin{align*}
S_g(f)&=S(g)^{-1}S(f+g'+(g-g'))=\\
&=S(g)^{-1}S(f+g')S(g')^{-1}S(g'+(g-g'))=\\
&=S(g)^{-1}S(g')S(g')^{-1}S(f+g')S_{g'}(g-g')=\\
&=\mathrm{Ad} S_{g'}(g-g')^{-1}(S_{g'}(f))\ .
\end{align*}
Hence the structure of local algebras depends only locally on the interaction. This allows to perform the  
adiabatic limit directly on the level of local algebras.

In the next step we want to remove the restriction to interactions with compact support. Let $G:\MM\to\RR^n$ be smooth and $\cO$ be bounded. Set
\begin{equation}\nonumber
[G]_\cO=\{g\in\cD^n|g\equiv G \text{ on a neighborhood of }J_+(\cO)\cap J_-(\cO)\} \ .
\end{equation}
We consider the $\tilde{\fA}$-valued maps
\begin{equation}\nonumber
S_{G,{\cO}}(f):[G]_{\cO}\ni g\mapsto S_g(f)\in\tilde{\fA}\ .
\end{equation}
The local algebra {$\fA_{G}(\cO)$} is defined to be the algebra generated by {$S_{G,{\cO}}(f),\supp f\subset \cO$}. Note that the evaluation maps
\begin{equation}\nonumber
\gamma_{gG}:S_{G,{\cO}}(f)\to S_g(f)
\end{equation}
extend to isomorphisms of $\fA_{G}(\cO)$ and $\fA_g(\cO)$ for every $g\in[G]_{\cO}$.

The local net is now defined by  the embeddings $i_{\cO_2\cO_1}$ for $\cO_1\subset\cO_2$
\begin{equation}\nonumber
i_{\cO_2\cO_1}:S_{G,{\cO_1}}(f)\mapsto S_{G,{\cO_2}}(f)
\end{equation}
for $f\in\cD^n$ with $\supp f\subset \cO_1$.
Let {$\fA_G$} be the inductive limit with embeddings
\begin{equation}\nonumber
i_{\cO}:\fA_{G}(\cO)\to\fA_G
\end{equation}
and we set {$$S_G(f)=i_{\cO}(S_{G,{\cO}}(f)).$$}
We are now ready to prove a crucial theorem about the net $\Ocal\mapsto\fA_{G}(\cO)$.
\begin{theorem} Let $G$ be translation invariant. Then the net becomes translation covariant by setting
\begin{equation}\nonumber
\alpha^G_x(S_G(f))=S_G(f_x)\ .
\end{equation}
\end{theorem}
\begin{proof}
We have to prove that $\al^G_x$ extends to an isomorphism from $\fA_G(\cO)\to\fA_G(\cO+x)$. Let $\cO_1\supset \cO\cup\cO-x$ and $g\in[G]_{\cO_1}$. 
Then $g,g_x\in[G]_{\cO}$ and $g_x=g+h^x_++h^x_-$ with $\supp h^x_\pm\cap J_\mp(\cO)=\varnothing$.
 By causal factorization
$$\al^G_x=\gamma_{gG}^{-1}\circ\mathrm{Ad}U_g(x)\circ\al_x\circ\gamma_{gG}$$
with
$U_g(x)=S_g(h^x_-)$. \qed
\end{proof}

\section{Time-slice axiom, operator product expansions, and the renormalization group}\label{sec::Time-slice}
We have seen that, starting from a free QFT and a definition of a time ordered product satisfying the axioms of Section \ref{sec:Time ordered products} we can construct a local net (in the sense of formal power series) satisfying the Haag-Kastler axioms of Isotony, Locality and Covariance. In this section we want to analyze the net in more detail.

First we investigate whether the net satisfies the time-slice axiom. This can be done for the case that the net is defined on a generic Lorentzian globally hyperbolic spacetime $M$. It is known since a long time \cite{FNW81} that the free theory generated by linear functionals, modulo the ideal of the free field equation, satisfies this axiom, and by using the techniques of microlocal analysis, this result can be extended to the net $\fA$ generated by elements of the form $\al_H^{-1}F$, where $F\in\Fcal_{\mc}$ is a microcausal functional \cite{HW01,CF08}. In \cite{CF08} it was  shown that this implies that also the net $\fA_G$ introduced in the previous section satisfies the axiom. The argument relies on the fact that the algebra of the interacting theory associated to some bounded region can be constructed as a subalgebra of the free theory for a slightly larger region, and vice versa.

The problem is that these subalgebras are fixed only up to unitary equivalence, so one has in addition to show that these unitary transformations can be appropriately fixed. We use the fact that the relative S-matrices $S_g(f)$ are well defined also for test functions $g$ with non-compact support provided the support is past compact, i.e. $\supp g\cap J_-(x)$ is compact for all $x\in M$. 

Let $\Sigma$ be a Cauchy surface of $M$ and $N$ a neighborhood of $\Sigma$. Let $\cO\subset M$ be relatively compact. We choose a Cauchy surface $\Sigma_-$ such that $\cO\cup N\subset J_+(\Sigma_-)$ and a smooth function $\chi$ with past compact support such that $\supp (1-\chi)\subset J_-(\Sigma_-)$. We want to prove that $\fA_{G\chi}(\cO)\subset\fA_{G\chi}(N)$. 

By construction of the interacting theory we see immediately that $\fA_{G\chi}(\cO)\subset \fA(M)$ holds. Due to the time slice property of the free theory, $\fA(M)=\fA(N')$ for each neighborhood $N'$ of $\Sigma$. We now construct within the algebra $\fA_{G\chi}(N)$ an algebra which is isomorphic to $\fA(N')$ for a sufficiently small $\Sigma\subset N'\subset N$.
For this purpose we choose another smooth function $\chi'$ with support contained in $J_+(N)$ and with $\supp (1-\chi')\subset J_-(N')$. Let now $\supp f\subset N'$. Then the unitaries
\begin{equation}
S_{G(\chi-\chi')}(f)=S_{G\chi}(g')^{-1}S_{G\chi}(g'+f)\ ,\  {\rm with}\ g'\equiv G\chi'\ {\rm on}\ J_-(\supp f), \supp g'\subset N\ , 
\end{equation}
generate an algebra isomorphic to $\fA(N')$ within $\fA_{G\chi}(N)$. The map
\begin{equation}
\alpha:S(f)\to S_{G(\chi-\chi')}(f)=\mathrm{Ad}(S(g-g'))^{-1}(S(f))
\end{equation}
with $g\equiv G\chi$ on $J_-(\supp f\cup\supp g')$ extends to an injective homomorphism from $\fA(M)$ into $\fA_{G\chi}(N)$. Since $\fA_{G\chi}(N)\subset \fA(M)$, $\alpha$ is an endomorphism of $\fA(M)$. We show that it is even an automorphism. For this purpose we construct the inverse of $\alpha$. By exploiting the time slice property of the free theory, we  
can restrict ourselves to elements $S(f)$ with $\supp f\subset J_-(\Sigma_0)$. On these elements we have
\begin{equation}
\alpha^{-1}(S(f))=\mathrm{Ad}(S(g-g'))(S(f))=S(g-g'+f)S(g-g')^{-1}
\end{equation}
where $g-g'\equiv G(\chi-\chi')$ on $J_+(\supp f)$. We conclude that $\fA_{G\chi}(N)=\fA(M)$. This proves the claim.


Another general property of the interacting net is the existence of an operator product expansion \cite{Hollands07}. In the case of the product of two fields $A$ and $B$ it is an expansion
\begin{equation}
A(x)B(y)\sim\sum_k C_{AB}^k(x,y)\ph_k(x)
\end{equation}
with distributions $C_{AB}^k$ and a basis of local fields $\ph_k$, ordered with respect to the scaling dimension.  This is an asymptotic expansion in the sense that after evaluation in a state coming from a Hadamard state of the free field, the difference between the right hand side of the relation and the left hand side, truncated at some $k$,   tends to zero as $x\to y$, with  an order depending on  $k$.  

The third property we look at is the behavior of the theory at different scales. In the standard formalism of QFT, one formulates this as a property of vacuum expectation values of products or time ordered products of fields, or one uses the concept of the so-called effective action. In this formulation one has to have control over the existence and uniqueness of the vacuum state. In the algebraic approach one can instead derive a relation between local nets. Namely given a local net $\cO\mapsto\fA_1(\cO)$ one obtains another net by scaling the regions,
\begin{equation}
\fA_{\lambda}(\cO)=\fA_{1}(\lambda\cO)\ .
\end{equation}
If the net depends on some parameters $(m,g)$, one can compensate the scaling by changing the parameters. One obtains
the algebraic Callan-Symanzik equation \cite{BDF09}
\begin{equation}
\fA_{\lambda}^{m,g}\cong\fA_1^{m(\lambda),g(\lambda)}
\end{equation}
The ``running'' of the parameters is as usual determined by the renormalization group equation which follows from the behavior of the time ordered product under scaling.

\section{Hamiltonian formalism for quantum field theory, and the construction of states}\label{Section:Hamiltonian}
Up to now we remained in the realm of algebras. There we could study several structural properties of the theory. In order to get more detailed predictions of the theory one has to evaluate the algebra in specific states.
A class of states on the local algebras can be obtained in terms of the states of the free theory by embedding the interacting theory into the free one, but this is highly ambiguous and gives no direct interpretation of the states. Conceptually, one does not need more, since the interpretation can be done in terms of the expectation values of observables. In practice, however, one would prefer to have states with an a priori interpretation as e.g. the vacuum state. The standard way to compute it is the evaluation of the product or the time ordered product of interacting fields with an interaction $L_I(g_0)=\int \cL_I(x)g_0(x)d\mu(x)$ in the vacuum state of the free theory and performing the adiabatic limit $g_0\to 1$. This limit is well behaved in massive theories, but exists also for a suitable sequence $(g_0)_n\to 1$ in certain massless theories such as massless $\ph^4$ or QED. In the case of time ordered products one just reproduces the standard formulas in terms of Feynman graphs; in the case of operator products one has to use Steinmann's sector graphs \cite{Steinmann92}. The adiabatic limit in this form, however, does not always exist, in particular not for states with nonzero temperature. 

A more direct way of constructing states with specific properties could be imagined in a Hamiltonian formalism, as well known from nonrelativistic quantum mechanics. The difficulty is that the interaction Hamiltonian for a local QFT is very singular so that perturbation theory for selfadjoint operators cannot be used. There are two independent reasons for the singular character of perturbations in QFT. The first is translation symmetry. In Minkowski space this leads to Haag's theorem, which states that the ground state of the interacting theory cannot be represented by a vector in the Fock space of the free theory. If one takes this into account by restricting the interaction to a finitely extended spatial region, one can indeed apply the perturbation theory of selfadjoint operators in certain superrenormalizable models in 2 dimensions. One can then construct ground states and consider their limit if the cutoff is removed. In 4 dimensions, however, the local interaction densities are too singular, so that also the spatially restricted interaction is not an operator. 

The Hamiltonian formalism relies on a split of spacetime into the product of a  Cauchy surface and the time axis, and all the observables of the theory are constructed in terms of their initial values on this surface, which are supposed to be independent of the interaction. But from renormalization theory it is well known that in general one has to expect modifications of the canonical structure; moreover, even for free fields, the restriction to a Cauchy surface is singular for all nonlinear local fields. 

Instead we use the fact that for generic perturbative QFT's the time-slice axiom holds. Moreover, as we saw from the discussion of the proof of this fact, the free and the interacting algebra of a time slice can be identified. This suggests to compare their time evolutions. Both are automorphism groups acting on the same algebra, and they differ by a cocycle. In case of a spatial cutoff of the interaction, the cocycle is implemented by a unitary cocycle within the algebra, whose generator is an integral over an operator valued function which may be interpreted as a regularized interaction Hamiltonian density $\cH_I(\x)$.  

As in section \ref{sec:Time ordered products} we consider the space $\Dcal^n$ of test functions and the algebra generated by $S(f)=\Scal(\al^{-1}_{\sst H}(\sum_i\int A_if^i d\mu))$. We also assume that $A_0=\cL_I$ is the interaction Lagrangian density. The time slice property proven in section \ref{sec::Time-slice} induces isomorphisms between the free and the interacting algebras. Let $\chi$ be a smooth function of time $t$ with $\chi(t)=1$ for $t>-\ep$ and $\chi(t)=0$ for $t\le -2\ep$. Then $\supp( (t,\x)\mapsto G(t,\x)\chi(t))$
is past compact. We now define a map from $\fA_G$ to $\fA$ by
\begin{equation}\nonumber
\gamma_{\chi}(S_G(f))=S_{G\chi}(f)\ ,\ \supp f\subset (-\ep,\ep)\times \RR^3\ .
\end{equation}
Due to the time slice property this map extends to an isomorphism. Moreover, it only slightly changes the kinematical localization at $t=0$. Let $\cO_r=\{(t,\x)||t|+|\x|<r\}$. Then
\begin{equation}\nonumber
\gamma_{\chi}(\fA_G(\cO_r))\subset\fA(\cO_{r+4\ep})\subset\gamma_{\chi}(\fA_G(\cO_{r+8\ep}))\ .
\end{equation}

Let $G$ be constant, 
let $\al_x^{G,\chi}=\gamma_{\chi}\circ\al_x^G\circ\gamma_{\chi^{-1}}$ be the translations of the interacting theory mapped to the free theory, and consider the cocycle $\beta^{G,\chi}_x=\alpha_x^{G,\chi}\circ\alpha_{-x}$. We find $\beta^{G,\chi}_{(0,\x)}=\mathrm{id}$ and, for $f$ with $\supp f\subset\cO_r$ and small $t$,
\begin{equation}\nonumber
\beta^{G,\chi}_{(t,0)}(S(f))=\mathrm{Ad}S_{h\chi}(h(\chi_t-\chi))(S(f))
\end{equation}
where $h$ is time independent, has compact spatial support and $h\equiv G$ on $\cO_{r+4\ep}$.
\begin{proposition}
The unitaries $U_t^{h\chi}=S_{h\chi}(h(\chi_t-\chi))$ fulfill the cocycle equation
\begin{equation}\nonumber
U_{t+s}^{h\chi}=U_t^{h\chi}\alpha_t(U_s^{h\chi})
\end{equation}
\end{proposition}
\begin{proof}
For sufficiently large $u$ (depending on $s,t$) we have
\begin{eqnarray*}
&S_{h\chi}(h(\chi_t-\chi))\al_t(S_{h\chi}(h(\chi_s-\chi))=\\
&S_{h(\chi-\chi_u)}(h(\chi_t-\chi))\al_t(S_{h(\chi-\chi_{u-t})}(h(\chi_s-\chi))=\\
&S_{h(\chi-\chi_u)}(h(\chi_t-\chi))S_{h(\chi_t-\chi_u)}(h(\chi_{t+s}-\chi_t))=\\
&S(h(\chi-\chi_u))^{-1}S(h(\chi_t-\chi_u))S(h(\chi_t-\chi_u))^{-1}S(h(\chi_{t+s}-\chi_u))\\
&=S_{h(\chi-\chi_u)}(h(\chi_{t+s}-\chi))=S_{h\chi}(h(\chi_{t+s}-\chi))\ .
\end{eqnarray*}
\qed
\end{proof}
We conclude that the unitary cocycle $U_t^{h\chi}$ describes the interacting time evolution (with spatial cutoff $h$) in the interaction picture. Due to the finite speed of propagation, it coincides with the full time evolution for small $t$. 

We now consider a time translation covariant representation $(\cH,\pi,U_0)$ and assume that the map $\cD\ni f\to \pi(S(f))$ is strongly continuous. Then the cocycle $U_t^{h\chi}$ is strongly continuous, and 
\begin{equation}
{t\mapsto U_{h\chi}(t)=U_t^{h\chi}U_0(t)}
\end{equation}
is a strongly continuous 1-parameter group with selfadjoint generater {$H_{h\chi}$} which describes the dynamics of the interacting system with spatial cutoff.

In case $\pi$ is irreducible, one may now determine the spectrum of $H_{h\chi}$  and interpret it as the energy spectrum of the interacting theory with spatial cutoff (up to an additive constant).
One may also look for a ground state and consider the limit of removal of the cutoff.

If $\pi$ is a representation induced by a KMS state, and $\Omega_0$ is the corresponding cyclic vector in the representation space, one knows by Connes' cocycle theorem that there exists a weight whose modular automorphims are the time translations of the interacting theory.  If $\Omega_0$ is in the domain of $e^{-\frac{\beta}{2}H_{h\chi}}$, then this weight is bounded and induced by the vector
\begin{equation}
\Omega_{h\chi}=e^{-\frac{\beta}{2}H_{h\chi}}\Omega_0  \ .
\end{equation}

If the cocycle is strongly differentiable on a dense domain, the interaction Hamiltonian can be defined as the generator of the cocycle.
We obtain \cite{FL14}
 \begin{equation}\nonumber
H_{I}^{h\chi}=\hbar\frac{d}{idt}U_t^{h\chi}=\hbar\, S(h\chi)^{-1}\frac{d}{idt}S(h\chi_t)=\!R_{V(h\chi)}(V(h\dot{\chi}))\,,
\end{equation}
where in the last step we have used Bogoliubov's formula \eqref{RV} for interacting fields. In the limit $\ep\to0$, $\dot{\chi}$ tends to the $\de$-function and we obtain the usual interaction picture.

We illustrate the method on the example of an interaction with external sources.
We start with the CCR algebra $\tilde{\fA}_{S_0}$ of the free scalar field introduced at the end of section \ref{Def:quant}. The
formal S-matrix is
\begin{equation}\nonumber
S(f)=\Scal(F_f)=e^{iF_f/\hbar}e^{-\frac{i}{2\hbar}\langle f,\Delta^D f\rangle}\,,
\end{equation}
where $F_f(\ph)=\int \ph f d\mu$. One can verify it by direct computation (using the forumlas for time-ordered product given in section \ref{explicit: constr}) or, indirectly, by the verification of the causal factorization property \eqref{causalfact}. Namely, we have
\begin{equation}
S(f+g)^{-1}S(f+g+h)=
\end{equation}
\begin{equation}\nonumber
e^{iF_{h}/\hbar}\exp\frac{i}{\hbar}(\langle f+g,\Delta^D (f+g)\rangle-\langle f+g+h,\Delta^D(f+g+h)\rangle+\langle f+g,\Delta h\rangle)
\end{equation}
$$=e^{i\ph(h)}\exp\frac{i}{2\hbar}(-\langle h,\Delta^D h\rangle-\langle f+g,(\underbrace{2\Delta^D-\Delta}_{=\Delta^A}) h\rangle) \ ,$$
$$=S(g)^{-1} S(g+h)e^{\frac{i}{2\hbar}\langle f,\Delta^A h\rangle}$$
hence if $\supp f\cap J_-(\supp h)=\varnothing$ then by the support property of the advanced propagator
$$\langle f,\Delta^A h\rangle=0$$
and the factorization holds.

We find the interaction Hamiltonian ($h$ time independent)
\begin{equation}\nonumber
H^{h\chi}_{I}=-\ph(h\dot{\chi})-\mathrm{const}\ .
\end{equation}
Due to the smearing in time, this operator remains meaningful also for a pointlike source ($h\sim\delta(\x)$).

In general, for the free theory we obtain the usual Fock space Hamiltonian $H_0$, and the Hamiltonian of the  interacting theory with spatial cutoff is the sum of the free Hamiltonian and the interaction term,
\begin{equation}
H=H_0+\int h(x)\cH_I({\bf x}) d^3\x \ .
\end{equation}
In this framework, one can now apply the standard perturbative constructions of ground states and KMS states. In \cite{Lindner13,FL14} it was shown that in massive theories ground states  and KMS states for positive temperatures exist. Some aspects of this formalism involving thermal mass were further developped in \cite{DHP} in conjunction with the principle of perturbative agreement \cite{HW05}. It is hoped that this regularized Hamiltonian picture will allow to close the conceptual gap between the standard formalism in nonrelativistic quantum mechanics and quantum statistical mechanics and the formalism of relativistic QFT. 
\section{Conclusions}
We have seen that the concepts of AQFT can be used in renormalized perturbative QFT and yield Haag-Kastler nets (in the sense of algebras of formal power series) for generic models of QFT. Due to its axiomatic formulation all possible renormalization methods are covered, and one has an a priori characterization of the class of renormalized theories associated to a classical Lagrangian, independent of any regularization scheme. For practical purposes, it is nevertheless often appropriate to introduce a regularization, and in particular analytic regularization schemes such as dimensional or analytic renormalization are useful, for computation but also for specifying a theory in its class (e.g. by minimal subtraction), see \cite{DFKR14}. One may also incorporate the ideas of the renormalization flow equation in the sense of Polchinski and made rigorous in \cite{KKS92} . This is exposed in \cite{BDF09}. In these notes we restricted ourselves to scalar field theories. The generalization to other types of field theory have been discussed in several papers; fermionic theories can be treated essentially in the same way, and gauge theories can be treated after adding auxiliary fields (ghosts etc.) and constructing the time ordered products such that BRST symmetry is respected \cite{DF99,Hollands07,FR13}. Even gravity can be included where however the concept of local algebras of observables has to be properly adapted \cite{BFR13}.

\end{document}